\newif\ifnotanonymous \notanonymoustrue
\newif\iffullversion \fullversiontrue
\documentclass[runningheads]{llncs}
\usepackage{graphicx}
\usepackage{color}
\usepackage{amsmath}
\usepackage{amssymb}
\usepackage{bcprules, proof}
\usepackage{fancybox}
\usepackage{mathtools}
\usepackage{float}
\usepackage{xparse}
\usepackage{lscape}
\usepackage{xspace}
\usepackage{url}
\usepackage{bcpproof}

\usepackage{cite}
\usepackage{hyperref}

\newcommand{\LTP}{$\lambda^{\triangleright\%}$\xspace}
\newcommand{\LMD}{$\lambda^{\textrm{MD}}$\xspace}
\newcommand{\LLF}{$\lambda\textrm{LF}$\xspace}

\newcommand{\G}{\Gamma}
\newcommand{\D}{\Delta}
\newcommand{\V}{\vdash_\Sigma}

\newcommand{\iskind}{\text{\ kind}}
\newcommand{\TW}{{\mathop{\triangleright}}}

\newcommand{\F}{\forall}
\newcommand{\TB}{{\mathop{\blacktriangleright}}}
\newcommand{\TBL}{{\mathop{\blacktriangleleft}}}
\newcommand{\E}{\equiv}
\newcommand{\FV}{\text{FV}}
\newcommand{\FTV}{\text{FSV}}

\newcommand{\WStar}{\textsc{W-Star}}
\newcommand{\WAbs}{\textsc{W-Abs}}

\newcommand{\KTConst}{\textsc{K-TConst}}
\newcommand{\KAbs}{\textsc{K-Abs}}
\newcommand{\KApp}{\textsc{K-App}}
\newcommand{\KConv}{\textsc{K-Conv}}
\newcommand{\KTW}{\textsc{K-$\TW$}}

\newcommand{\KGen}{\textsc{K-Gen}}
\newcommand{\KCsp}{\textsc{K-Csp}}

\newcommand{\TConst}{\textsc{T-Const}}
\newcommand{\TVar}{\textsc{T-Var}}
\newcommand{\TAbs}{\textsc{T-Abs}}
\newcommand{\TApp}{\textsc{T-App}}
\newcommand{\TConv}{\textsc{T-Conv}}
\newcommand{\TTB}{\textsc{T-$\TB$}}
\newcommand{\TTBL}{\textsc{T-$\TBL$}}
\newcommand{\TGen}{\textsc{T-Gen}}
\newcommand{\TIns}{\textsc{T-Ins}}
\newcommand{\TCsp}{\textsc{T-Csp}}

\newcommand{\QKAbs}{\textsc{QK-Abs}}

\newcommand{\QKRefl}{\textsc{QK-Refl}}
\newcommand{\QKSym}{\textsc{QK-Sym}}
\newcommand{\QKTrans}{\textsc{QK-Trans}}

\newcommand{\QTAbs}{\textsc{QT-Abs}}
\newcommand{\QTApp}{\textsc{QT-App}}
\newcommand{\QTTW}{\textsc{QT-$\TW$}}
\newcommand{\QTGen}{\textsc{QT-Gen}}
\newcommand{\QTCsp}{\textsc{QT-Csp}}
\newcommand{\QTRefl}{\textsc{QT-Refl}}
\newcommand{\QTSym}{\textsc{QT-Sym}}
\newcommand{\QTTrans}{\textsc{QT-Trans}}

\newcommand{\QAbs}{\textsc{Q-Abs}}
\newcommand{\QApp}{\textsc{Q-App}}
\newcommand{\QTB}{\textsc{Q-$\TB$}}
\newcommand{\QTBL}{\textsc{Q-$\TBL$}}
\newcommand{\QGen}{\textsc{Q-Gen}}
\newcommand{\QIns}{\textsc{Q-Ins}}
\newcommand{\QCsp}{\textsc{Q-Csp}}
\newcommand{\QRefl}{\textsc{Q-Refl}}
\newcommand{\QSym}{\textsc{Q-Sym}}
\newcommand{\QTrans}{\textsc{Q-Trans}}
\newcommand{\QBeta}{\textsc{Q-$\beta$}}

\newcommand{\QTBLTB}{\textsc{Q-$\TBL\TB$}}
\newcommand{\QLambda}{\textsc{Q-$\Lambda$}}
\newcommand{\QPercent}{\textsc{Q-\%}}

\newcommand{\I}{\textrm{Int}}
\newcommand{\B}{\textrm{Bool}}

\newcommand{\rulefbox}[1]{\fbox{\ensuremath{#1}} \hspace{1mm}}

\begin{document}
\title{A Dependently Typed Multi-Stage Calculus%
}
%
%
\ifnotanonymous
	\author{Akira Kawata\inst{1} \and
		Atsushi Igarashi\inst{1}\orcidID{0000-0002-5143-9764}}
	\authorrunning{A. Kawata, A. Igarashi}
	\institute{Graduate School of Informatics, Kyoto University, Kyoto, Japan\\
          \email{akira@fos.kuis.kyoto-u.ac.jp} \\
          \email{igarashi@kuis.kyoto-u.ac.jp}
	}
\fi
\maketitle              
\begin{abstract}

	We study a dependently typed extension of a multi-stage programming language
\`a la MetaOCaml, which supports quasi-quotation and cross-stage persistence
for manipulation of code fragments as first-class values and an evaluation
construct for execution of programs dynamically generated by this code
manipulation. Dependent types are expected to bring to multi-stage
programming enforcement of strong invariant---beyond simple type safety---on
the behavior of dynamically generated code. An extension is, however, not
trivial because such a type system would have to take stages of types---roughly speaking,
the number of surrounding quotations---into account.

To rigorously study properties of such an extension, we develop \LMD, which
is an extension of Hanada and Igarashi's typed calculus \LTP with dependent
types, and prove its properties including preservation, confluence, strong
normalization for full reduction, and progress for staged reduction.
Motivated by code generators that generate code whose type depends on
a value from outside of the quotations, we argue the significance of cross-stage
persistence in dependently typed multi-stage programming and certain type
equivalences that are not directly derived from reduction rules.

	
\end{abstract}


\section{Introduction}
\label{sec:intro}

\subsection{Multi-stage Programming and MetaOCaml}

Multi-stage programming makes it easier for programmers to implement
generation and execution of code at run time by providing language
constructs for composing and running pieces of code as first-class
values.  A promising application of multi-stage programming is
(run-time) code specialization, which generates program code
specialized to partial inputs to the program and such applications are
studied in the literature~\cite{8384206,mainland2012metahaskell,taha2007gentle}.

MetaOCaml~\cite{calcagno2003implementing,oleg2014} is an extension of
OCaml\footnote{\url{http://ocaml.org}} with special constructs for
multi-stage programming, including brackets and escape, which are
(hygienic) quasi-quotation, and \texttt{run}, which is similar to
\texttt{eval} in Lisp, and cross-stage persistence
(CSP)~\cite{MetaML}.  Programmers can easily write code generators by
using these features.  Moreover, MetaOCaml is equipped with a powerful
type system for safe code generation and execution.  The notion of
code types is introduced to prevent code values that represent
ill-typed expressions from being generated.  For example, a quotation
of expression \texttt{1 + 1} is given type \texttt{int code}
and a code-generating function, which takes a code value \(c\) as an
argument and returns \(c \texttt{ + } c\), is given type \texttt{int
  code -> int code} so that it cannot be applied to, say, a quotation
of \texttt{"Hello"}, which is given type \texttt{string
  code}.  Ensuring safety for \verb|run| is more challenging because
code types by themselves do not guarantee that the execution of code values never
results in unbound variable errors.  Taha and
Nielsen~\cite{taha2003environment} introduced the notion of
environment classifiers to address the problem, developed a type
system to ensure not only type-safe composition but also type-safe
execution of code values, and proved a type soundness theorem (for a formal calculus \(\lambda^\alpha\) modeling a pure subset of MetaOCaml).

However, the type system, which is based on the Hindley--Milner
polymorphism~\cite{Milner78JCSS}, is not strong enough to guarantee
invariant beyond simple types.  For example, Kiselyov~\cite{8384206}
demonstrates specialization of vector/matrix computation with respect
to the sizes of vectors and matrices in MetaOCaml but the type system
of MetaOCaml cannot prevent such specialized functions from being applied to
vectors and matrices of different sizes.




\subsection{Multi-stage Programming with Dependent Types}


One natural idea to address this problem is the introduction of dependent types
to express the size of data structures in static types~\cite{Xi98}.
For example, we could declare vector types indexed by the size of
vectors as follows.
\begin{verbatim}
    Vector :: Int -> *
\end{verbatim}
\verb|Vector| is a type constructor that takes an integer (which
represents the length of vectors): for example, \verb|Vector 3| is the
type for vectors whose lengths are 3.  Then, our hope is to specialize
vector/matrix functions with respect to their size and get a piece of
function code, whose type respects the given size, \emph{provided at
  specialization time}.  For example, we would like to specialize a
function to add two vectors with respect to the size of vectors, that
is, to implement a code generator that takes a (nonnegative) integer $n$ as an
input and generates a piece of function code of type
\verb|(Vector |$n$\verb| -> Vector |$n$\verb| -> Vector |$n$\verb|) code|.


\subsection{Our Work}
In this paper, we develop a new multi-stage calculus \LMD by extending
the existing multi-stage calculus \LTP\cite{Hanada2014} with dependent
types and study its properties.  We base our work on \LTP, in which
the four multi-stage constructs are handled slightly differently from
MetaOCaml, because its type system and semantics are arguably simpler
than \(\lambda^\alpha\)~\cite{taha2003environment}, which formalizes
the design of MetaOCaml more faithfully.  Dependent types are based on
\LLF~\cite{attapl}, which has one of the simplest forms of dependent
types.  Our technical contributions are summarized as follows:
\begin{itemize}
\item We give a formal definition of \LMD with its syntax, type system and
two kinds of reduction: full reduction, allowing reduction of any redex,
including one under $\lambda$-abstraction and quotation, and staged reduction, a
small-step call-by-value operational semantics that is closer to the intended
multi-stage implementation.
\item We show preservation, strong normalization, and confluence for
  full reduction; and show unique decomposition (and progress as its
  corollary) for staged reduction.
\end{itemize}
The combination of multi-stage programming and dependent types has
been discussed by Pasalic, Taha, and Sheard~\cite{pasalic2002tagless}
and Brady and Hammond~\cite{brady2006dependently} but, to our
knowledge, our work is a first formal calculus of full-spectrum dependently typed
multi-stage programming with all the key constructs mentioned above.

\subsubsection{Organization of the Paper.}

The organization of this paper is as follows.
Section~\ref{sec:informal-overview} gives an informal overview of
\LMD.  Section~\ref{sec:formal-definition} defines \LMD and
Section~\ref{sec:properties} shows properties of \LMD.
Section~\ref{sec:related-work} discusses related work and Section
\ref{sec:conclusion} concludes the paper with discussion of future
work.  \iffullversion
For reference, we give the full definition of \LMD and proofs of properties in the Appendix.
\fi


\section{Informal Overview of \LMD \label{sec:informal-overview}}

We describe our calculus \LMD informally.  \LMD is based on
\LTP~\cite{Hanada2014} by Hanada and Igarashi and so we start with a review of 
\LTP.

\subsection{\LTP}


In \LTP, brackets (quasi-quotation) and escape (unquote) are written
$\TB_\alpha M$ and $\TBL_\alpha M$, respectively.  For example,
$\TB_\alpha (1 + 1)$ represents code of expression $1 + 1$ and thus
evaluates to itself.  Escape $\TBL_\alpha M$ may appear under
$\TB_\alpha$; it evaluates $M$ to a code value and splices it into the
surrounding code fragment.  Such splicing is expressed by the
following reduction rule:
\begin{align*}
	\TBL_\alpha (\TB_\alpha M) \longrightarrow M .
\end{align*}

The subscript $\alpha$ in $\TB_\alpha$ and $\TBL_\alpha$ is a \textit{stage
  variable}\footnote{%
  In Hanada and Igarashi~\cite{Hanada2014}, it was called a
  \textit{transition variable}, which is derived from correspondence
  to modal logic, studied by Tsukada and Igarashi~\cite{Tsukada}.} and
a sequence of stage variables is called a \textit{stage}.  Intuitively,
a stage represents the depth of nested brackets.  Stage variables can be
abstracted by $\Lambda\alpha.M$ and instantiated by an application
$M\ A$ to stages.  For example,
$\Lambda\alpha.\TB_\alpha ((\lambda x:\I.x+10)\ 5)$ is a code value,
where \(\alpha\) is abstracted.  If it is applied to \(A = \alpha_1 \cdots \alpha_n\), \(\TB_\alpha\) becomes \(\TB_{\alpha_1} \cdots \TB_{\alpha_n}\); in particular,
if \(n = 0\), \(\TB_\alpha\) disappears.  So, an
application of $\Lambda\alpha.\TB_\alpha ((\lambda x:\I.x+10)\ 5)$
to the empty sequence \(\varepsilon\) reduces to
(unquoted) \((\lambda x:\I.x+10)\ 5\) and to 15.  In other words, application of a
\(\Lambda\)-abstraction to $\varepsilon$ corresponds to \texttt{run}.
This is expressed by the following reduction rule:
\begin{align*}
	(\Lambda\alpha.M)\ A \longrightarrow M[\alpha\mapsto A]
\end{align*}
where stage substitution \([\alpha \mapsto A]\) manipulates the nesting of
\(\TB_\alpha\) and \(\TBL_\alpha\) (and also \(\%_\alpha\) as we see later).

Cross-stage persistence (CSP), which is an important feature of \LTP,
is a primitive to embed values (not necessarily code values) into a
code value.  For example, a \LTP-term
\[
  M_1 = \lambda x:\I.\Lambda\alpha.(\TB_\alpha ((\%_\alpha x) * 2))
\]
takes an integer \(x\) as an input and returns a code value, into
which \(x\) is embedded.  If $M_1$ is applied to $38 + 4$ as in
\[
  M_2 = (\lambda x:\I.\Lambda\alpha.(\TB_\alpha ((\%_\alpha x) * 2)))\ (38 + 4),
\]
then it evaluates to
\(M_3 = \Lambda\alpha.(\TB_\alpha ((\%_\alpha 42) * 2))\).  According
to the semantics of \LTP, the subterm $\%_\alpha 42$ means that it
waits for the surrounding code to be run (by an application to
$\varepsilon$) and so it does not reduce further.  If \(M_3\) is run
by application to \(\varepsilon\), substitution of \(\varepsilon\) for
\(\alpha\) eliminates \(\TB_\alpha\) and \(\%_\alpha\) and so
\(42 * 2\), which reduces to 84, is obtained.
CSP is practically important because
one can call library functions from inside quotations.  

The type system of \LTP uses code types---the type of code of type
\(\tau\) is written \(\TW_\alpha \tau\)---for typing \(\TB_\alpha\),
\(\TBL_\alpha\) and \(\%_\alpha\).  It takes stages into account: a
variable declaration (written $x:\tau@A$) in a type environment is associated with its
declared stage $A$ as well as its type $\tau$ and the type judgement of \LTP is of
the form $\G \vdash M : \tau@A$, in which $A$ stands for the stage
of term $M$.\footnote{%
  In Hanada and Igarashi~\cite{Hanada2014}, it is written
  $\G \vdash^A M : \tau$.
  }
For example,
$y:\I@\alpha \vdash (\lambda x:\I.y) : \I \to \I @ \alpha$ holds, but
$y:\I@\alpha \vdash (\lambda x:\I.y) : \I \to \I @ \varepsilon$ does
not because the latter uses $y$ at stage \(\varepsilon\) but $y$ is
declared at $\alpha$.  Quotation \(\TB_\alpha M\) is given type
\(\TW_\alpha \tau\) at stage \(A\) if \(M\) is given type \(\tau\) at
stage \(A\alpha\); unquote \(\TBL_\alpha M\) is given type \(\tau\)
at stage \(A\alpha\) if \(M\) is given type \(\TW_\alpha \tau\) at
stage \(A\alpha\); and CSP \(\%_\alpha M\) is give type \(\tau\)
at stage \(A\alpha\) if \(M\) is given type \(\tau\) at \(A\).
These are expressed by the following typing rules.
\begin{center}
	\infrule{\G\vdash M:\tau @{A\alpha}}{\G\vdash \TB_{\alpha}M:\TW_{\alpha}\tau @A} \hfil
	\infrule{\G\vdash M:\TW_{\alpha}\tau @A}{\G\vdash \TBL_{\alpha}M:\tau @{A\alpha}} \hfil
	\infrule{\G\vdash M: \tau @A}{\G\vdash \%_{\alpha}M:\tau @{A\alpha}}
\end{center}

\subsection{Extending \LTP with Dependent Types}

In this paper, we add a simple form of dependent types---{\`a} la
Edinburgh LF~\cite{harper1993framework} and \LLF~\cite{attapl}---to \LTP.
Types can be indexed by terms as in \texttt{Vector} in
Section~\ref{sec:intro} and \(\lambda\)-abstractions can be given
dependent function types of the form \(\Pi x:\tau. \sigma\) but we do
not consider type operators (such as $\texttt{list } \tau$) or
abstraction over type variables.  We introduce kinds to classify
well-formed types and equivalences for kinds, types, and terms---as
in other dependent type systems---but we have to address a question
how the notion of stage (should) interact with kinds and types.

On the one hand, base types such as \(\I\) should be able to be used
at every stage as in \LTP so that
\(\lambda x:\I.\Lambda \alpha. \TB_\alpha \lambda y:\I.M\) is a valid
term (here, \(\I\) is used at \(\varepsilon\) and \(\alpha\)).
Similarly for indexed types such as Vector 4.  On the other hand, it
is not immediately clear how a type indexed by a variable, which can be used only
at a declared stage, can be used.  For example,
consider
\[\TB_\alpha (\lambda x:\I. (\TBL_\alpha (\lambda y:\text{Vector
  }x.M)N) )
  \text{ and }
  \lambda x:\I. \TB_\alpha (\lambda y:\text{Vector }x.M).
\]
Is Vector\ \(x\) a legitimate type at \(\varepsilon\) (and \(\alpha\),
resp.)  even if \(x:\I\) is declared at stage \(\alpha\) (and
\(\varepsilon\), resp.)?  We will give our answer to this question in two
steps.

First, type-level constants such as \(\I\) and Vector can be used at
every stage in \LMD.  Technically, we introduce a signature that
declares kinds of type-level constants and types of constants.  For
example, a signature for the Boolean type and constants is given as
follows $\B::*, \text{true}:\B, \text{false}:\B$ (where $*$ is the
kind of proper types).  Declarations in a signature are not
associated to particular stages; so they can be used at every stage.

Second, an indexed type such as Vector\ 3 or Vector\ $x$ is well
formed only at the stage(s) where the index term is well-typed.  Since
constant \(3\) is well-typed at every stage (if it is declared in the
signature), Vector\ 3 is a well-formed type at every stage, too.
However, Vector\ $M$ is well-formed only at the stage where index term
$M$ is typed.  Thus, the kinding judgment
of \LMD takes the form \(\G\V \tau :: K @ A\), where stage $A$ stands for
where \(\tau\) is well-formed.  For example,
given \(\text{Vector}:: \I \rightarrow *\) in the signature \(\Sigma\),
\(x:\I@\varepsilon \V \text{Vector }x :: * @\varepsilon\) can be
derived but
neither
\(x:\I@\alpha \V \text{Vector }x :: * @\varepsilon\)
nor 
\(x:\I@\varepsilon \V \text{Vector }x :: * @\alpha\)
can be.

Apparently, the restriction above sounds too severe, because a term like
\(\lambda x:\I. \TB_\alpha (\lambda y:\text{Vector }x.M) \), which models a
typical code generator which takes the size $x$ and returns code for vector
manipulation specialized to the given size, will be rejected. It seems crucial
for \(y\) to be given a type indexed by $x$. We can address this problem by
CSP---In fact, $\text{Vector }x$ is not well formed at $\alpha$ under
$x:\I@\varepsilon$ but $\text{Vector }(\%_\alpha x)$ is!  Thus, we can still
write \(\lambda x:\I. \TB_\alpha (\lambda y:\text{Vector }(\%_\alpha x).M) \)
for the typical sort of code generators.

Our decision that well-formedness of types takes stages of index terms
into account will lead to the introduction of CSP at the type level
and special equivalence rules, as we will see later.


\section{Formal Definition of \LMD \label{sec:formal-definition}}
\label{sec:formal}

In this section, we give a formal definition of \LMD, including
the syntax, full reduction, and type system.  In addition to the full reduction,
in which any redex at any stage can be reduced, we also give staged reduction,
which models program execution (at \(\varepsilon\)-stage).

\subsection{Syntax}

We assume the denumerable set of \emph{type-level constants}, ranged over by
metavariables \(X, Y, Z\), the denumerable set of \emph{variables}, ranged
over by \(x,y,z\), the denumerable set of \emph{constants}, ranged over by
\(c\), and the denumerable set of \emph{stage variables}, ranged over by
\(\alpha, \beta, \gamma\).  The metavariables \(A, B, C\) range over
sequences of stage variables; we write \(\varepsilon\) for the empty
sequence. \LMD is defined by the following grammar:

{
\begin{align*}
    \textrm{kinds}             &  & K,J,I,H,G                & ::= * \mid \Pi x:\tau.K                                                           \\
    \textrm{types}             &  & \tau,\sigma,\rho,\pi,\xi & ::= X \mid \Pi x:\tau.\sigma \mid \tau\ M \mid \TW_{\alpha} \tau \mid \F\alpha.\tau \\
    \textrm{terms}             &  & M,N,L,O,P                & ::= c \mid x \mid \lambda x:\tau.M\ \mid M\ N \mid \TB_\alpha M                   \\
                               &  &                          & \ \ \ \ \mid \TBL_\alpha M \mid \Lambda\alpha.M \mid M\ A \mid \%_\alpha M        \\
    \textrm{signatures}         &  & \Sigma                   & ::= \emptyset \mid \Sigma, X::K \mid \Sigma, c:\tau                               \\
    \textrm{type env.} &  & \Gamma                   & ::= \emptyset \mid  \Gamma,x:\tau @A                                              \\
\end{align*}
}




A kind, which is used to classify types, is either $*$, the kind of
proper types (types that terms inhabit), or $\Pi x\colon\tau.K$, the kind
of type operators that takes $x$ as an argument of type $\tau$ and returns a type
of kind $K$.
A type is a type-level constant $X$, which is declared in the signature with its kind, a dependent function type $\Pi x:\tau.\sigma$,
an application $\tau\ M$ of a type (operator of $\Pi$-kind) to a term, a code type $\TW_\alpha \tau$, or an $\alpha$-closed type $\F\alpha.\tau$.
An example of an application of a type (operator) of $\Pi$-kind to a term is $\text{Vector}\ 10$; it is well kinded
if, say, the type-level constant $\text{Vector}$ has kind $\Pi x:\I.*$.
A code type $\TW_\alpha \tau$ is for a code fragment of a term of type $\tau$.
An $\alpha$-closed type, when used with $\TW_\alpha$, represents runnable code.



Terms include ordinary (explicitly typed) \(\lambda\)-terms, constants,
whose types are declared in signature $\Sigma$, and the following five forms
related to multi-stage programming:
$\TB_\alpha M$ represents a code fragment; $\TBL_\alpha M$ represents escape;
$\Lambda\alpha.M$ is a stage variable abstraction;
$M\ A$ is an application of a stage abstraction $M$ to stage $A$; and
$\%_\alpha M$ is an operator for cross-stage persistence.


We adopt the tradition of \LLF-like systems, where types of constants and
kinds of type-level constants are globally declared in a signature $\Sigma$,
which is a sequence of declarations of the form $c:\tau$ and $X::K$. For
example, when we use Boolean in \LMD, $\Sigma$ includes $\B :: *,
\textrm{true}:\B, \textrm{false}:\B$. Type environments are sequences of
triples of a variable, its type, and its stage. We write
\(\textit{dom}(\Sigma)\) and \(\textit{dom}(\Gamma)\) for the set of
(type-level) constants and variables declared in \(\Sigma\) and \(\Gamma\),
respectively. As in other multi-stage
calculi~\cite{taha2003environment,Tsukada,Hanada2014}, a variable declaration
is associated with a stage so that a variable can be referenced only at the
declared stage. On the contrary, constants and type-level constants are
\emph{not} associated with stages; so, they can appear at any stage. We
define well-formed signatures and well-formed type environments later.

The variable $x$ is bound in $M$ by $\lambda x:\tau.M$ and in $\sigma$
by $\Pi x:\tau.\sigma$, as usual; the stage variable $\alpha$ is
bound in $M$ by $\Lambda \alpha.M$ and $\tau$ by $\F\alpha.\tau$.
The notion of free variables is defined in a standard manner.
We write $\FV(M)$ and $\FTV(M)$ for the set of free variables and the set of free stage variables in $M$, respectively.  Similarly, $\FV(\tau)$, $\FTV(\tau)$,
$\FV(K)$, and $\FTV(K)$ are defined.
We sometimes abbreviate $\Pi x:\tau_1.\tau_2$ to $\tau_1 \rightarrow \tau_2$ if
$x$ is not a free variable of $\tau_2$.
We identify $\alpha$-convertible terms and assume the names of bound variables are pairwise distinct.

The prefix operators $\TW_\alpha, \TB_\alpha, \TBL_\alpha$, and
$\%_\alpha$ are given higher precedence over the three forms $\tau\ M$, $M\ N$,
$M\ A$ of applications, which are left-associative. The binders $\Pi$,
$\forall$, and $\lambda$ extend as far to the right as possible.
Thus, $\F\alpha.\TW_{\alpha} (\Pi x:\I.\text{Vector}\ 5)$ is
interpreted as
$\F\alpha.(\TW_{\alpha} (\Pi x:\I.(\text{Vector}\ 5)))$; and
$\Lambda\alpha.\lambda x:\I.\TB_\alpha x\ y$ means
$\Lambda\alpha.(\lambda x:\I.(\TB_\alpha x)\ y)$.

\paragraph{Remark:} Basically, we define \LMD to be an extension of
\LTP with dependent types.  One notable difference is that \LMD has
only one kind of \(\alpha\)-closed types, whereas \LTP has two kinds
of \(\alpha\)-closed types \(\forall\alpha.\tau\) and
\(\forall^\varepsilon\alpha.\tau\).  We have omitted the first kind,
for simplicity, and dropped the superscript $\varepsilon$ from the second. It
would not be difficult to recover the distinction to show properties related
to program residualization~\cite{Hanada2014}, although they are left as 
conjectures.

\subsection{Reduction}

Next, we define full reduction for \LMD.
Before giving the definition of reduction, we define two kinds of substitutions.
Substitution $M[x\mapsto N], \tau[x \mapsto N]$ and $K[x \mapsto N]$ are
ordinary capture-avoiding substitution of
term $N$ for $x$ in term $M$, type $\tau$, and kind $K$, respectively,
and we omit their definitions here.
Substitution $M[\alpha \mapsto A], \tau [\alpha \mapsto A], K[\alpha \mapsto A]$ and $B[\alpha\mapsto A]$ are
substitutions of stage $A$ for stage variable $\alpha$ in term $M$, type $\tau$, kind $K$, and stage $B$, respectively.
We show representative cases below.
{
\begin{align*}
    (\lambda x:\tau.M)[\alpha \mapsto A] & = \lambda x:(\tau[\alpha \mapsto A]).(M[\alpha \mapsto A])                                  \\
    (M\ B)[\alpha \mapsto A]             & = (M[\alpha \mapsto A])\ B[\alpha\mapsto A]                                                 \\
    (\TB_\beta M)[\alpha \mapsto A]      & = \TB_{\beta[\alpha \mapsto A]}M[\alpha \mapsto A]                                          \\
    (\TBL_\beta M)[\alpha \mapsto A]     & = \TBL_{\beta[\alpha \mapsto A]}M[\alpha \mapsto A]                                         \\
    (\%_\beta M)[\alpha \mapsto A]       & = \%_{\beta[\alpha \mapsto A]}M[\alpha \mapsto A]                                           \\
    (\beta B)[\alpha \mapsto A]          & = \beta (B[\alpha\mapsto A])                               & (\text{if } \alpha \neq \beta) \\
    (\beta B)[\alpha \mapsto A]          & = A (B[\alpha\mapsto A])                                   & (\text{if } \alpha = \beta)
\end{align*}
}
Here, $\TB_{\alpha_1\cdots\alpha_n} M$,
$\TBL_{\alpha_1\cdots\alpha_n} M$, and $\%_{\alpha_1\cdots\alpha_n} M$
$(n \geq 0)$ stand for $\TB_{\alpha_1} \cdots \TB_{\alpha_n} M$,
$\TBL_{\alpha_n}\cdots \TBL_{\alpha_1} M$, and
$\%_{\alpha_n}\cdots \%_{\alpha_1} M$, respectively.  
In particular,
$\TB_{\varepsilon} M = \TBL_{\varepsilon} M = \%_{\varepsilon} M = M$.
Also, it is important that
the order of stage variables is reversed for $\TBL$ and $\%$.
We also define substitutions of a stage or a term for variables in type environment $\G$.

\begin{definition}[Reduction]
    The relations $M \longrightarrow_\beta N$, $M \longrightarrow_\blacklozenge N$, and $M \longrightarrow_\Lambda N$
    are the least compatible relations closed under the rules below.
{
    \begin{align*}
         (\lambda x:\tau.M) N & \longrightarrow_\beta M[x \mapsto N]         \\
         \TBL_\alpha \TB_\alpha M & \longrightarrow_\blacklozenge M          \\
         (\Lambda \alpha.M)\ A & \longrightarrow_\Lambda M[\alpha \mapsto A]
    \end{align*}
  }    
\end{definition}
We write $ M \longrightarrow M'$ iff $ M \longrightarrow_\beta M'$,
$ M \longrightarrow_\blacklozenge M'$, or
$ M \longrightarrow_\Lambda M'$ and we call $\longrightarrow_\beta$,
$\longrightarrow_\blacklozenge$, and $\longrightarrow_\Lambda$
$\beta$-reduction, $\blacklozenge$-reduction, and $\Lambda$-reduction,
respectively.
$M \longrightarrow^* N$ means that there is a sequence of reduction $\longrightarrow$ whose length is greater than or equal to 0.

The relation $\longrightarrow_\beta$ represents ordinary $\beta$-reduction in the \(\lambda\)-calculus; the relation
$\longrightarrow_\blacklozenge$ represents that quotation $\TB_\alpha M$ is canceled by escape and $M$ is spliced into the code fragment surrounding the escape;
the relation $\longrightarrow_\Lambda$ means that a stage abstraction applied to  stage $A$ reduces to the body of the abstraction
where $A$ is substituted for the stage variable.
There is no reduction rule for CSP as with Hanada and Igarashi \cite{Hanada2014}.
The CSP operator $\%_\alpha$ disappears when $\varepsilon$ is substituted for $\alpha$.
We show an example of a reduction sequence below.
Underlines show the redexes.
\begin{align*}
     & \hspace{10mm} \underline{(\lambda f:\I\to\I.(\Lambda\alpha.\TB_\alpha (\%_\alpha f\ 1 + (\TBL_\alpha \TB_\alpha 3))\ \varepsilon))\ (\lambda x:\I.x)} \\
     & \longrightarrow_\beta (\Lambda\alpha.\TB_\alpha (\%_\alpha (\lambda x:\I.x)\ 1 + (\underline{\TBL_\alpha \TB_\alpha 3})))\ \varepsilon        \\
     & \longrightarrow_\blacklozenge \underline{(\Lambda\alpha.\TB_\alpha (\%_\alpha (\lambda x:\I.x)\ 1 + 3))\ \varepsilon}                                         \\
     & \longrightarrow_\Lambda \underline{(\lambda x:\I.x)\ 1} + 3                                                                                           \\
     & \longrightarrow_\beta 1 + 3                                                                                                                           \\
     & \longrightarrow^* 4
\end{align*}


\subsection{Type System}

In this section, we define the type system of \LMD.
It consists of eight judgment forms for signature well-formedness, type environment well-formedness, kind well-formedness, kinding, typing, kind equivalence, type equivalence, and term equivalence.
We list the judgment forms in Figure~\ref{fig:LMD-six-judgments}.
They are all defined in a mutual recursive manner.  We will discuss
each judgment below.

\begin{figure}[tbp]
  \begin{center}
    \begin{align*}
      \vdash & \Sigma                     & \text{signature well-formedness}        \\
      \V     & \G                         & \text{type environment well-formedness} \\
      \G     & \V K \iskind @ A           & \text{kind well-formedness}             \\
      \G     & \V \tau :: K @ A           & \text{kinding}                          \\
      \G     & \V M : \tau @ A            & \text{typing}                           \\
      \G     & \V K \E J @ A              & \text{kind equivalence}                 \\
      \G     & \V \tau \E \sigma :: K @ A & \text{type equivalence}                 \\
      \G     & \V M \E N : \tau @ A       & \text{term equivalence}
    \end{align*}
    \caption{Eight judgment forms of the type system of \LMD.}
    \label{fig:LMD-six-judgments}
  \end{center}
\end{figure}

\subsubsection{Signature and Type Environment Well-formedness.}
The rules for Well-formed signatures and type environments are
shown below:
{\small
\begin{center}
  \infrule{
  }{
    \vdash \emptyset
  }
  \hfil
  \infrule{
    \vdash \Sigma \andalso
    \V K \iskind @ \varepsilon \\
    X\notin\textit{dom}(\Sigma)
  }{
    \vdash \Sigma, X::K
  }
  \hfil
  \infrule{
    \vdash \Sigma \andalso
    \V \tau :: * @ \varepsilon \\
    c\notin\textit{dom}(\Sigma)
  }{
    \vdash \Sigma, c:\tau
  }
  \\[2mm]
  \infrule{
  }{
    \V \emptyset
  }
  \hfil
  \infrule{
    \V \Gamma \andalso
    \Gamma \V \tau :: * @ A \andalso
    x\notin\textit{dom}(\Sigma)
  }{
    \V \Gamma, x:\tau@A
  }
\end{center}
}

To add declarations to a signature, the kind/type of a (type-level)
constant has to be well-formed at stage \(\varepsilon\) so that it is
used at any stage.  In what follows, well-formedness is not explicitly
mentioned but we assume that all signatures and type environments are
well-formed.

\subsubsection{Kind Well-formedness and Kinding.}

The rules for kind well-formedness and kinding are a straightforward
adaptation from \LLF and \LTP, except for the following rule for type-level CSP.
\begin{center}
  \infrule[K-Csp]{
    \G \V \tau :: * @ A
  }{
    \G \V \tau :: * @ A\alpha
  }
\end{center}
Unlike the term level, type-level CSP is implicit because there is no staged
semantics for types.

\subsubsection{Typing.}

\begin{figure}[tbp]
  \begin{center}
    \infrule[\TConst]{c:\tau \in \Sigma}{\G \V c:\tau @A} \hfil
    \infrule[\TVar]{x:\tau @A \in \G}{\G \V x:\tau @A} \\[2mm]
    \infrule[\TAbs]{\G\V \sigma::*@A\andalso\G,x:\sigma@A\V M:\tau @A}{\G\V(\lambda (x:\sigma).M):(\Pi (x:\sigma).\tau)@A} \\[2mm]
    \infrule[\TApp]{\G\V M:(\Pi (x:\sigma).\tau)@A \andalso \G\V N:\sigma@A}{\G\V M\ N : \tau[x\mapsto N]@A} \\[2mm]
    \infrule[\TConv]{\G\V M:\tau @A \andalso \G\V \tau\equiv \sigma :: K@A}{\G\V M:\sigma@A} \\[2mm]
    \infrule[\TTB]{\G\V M:\tau @{A\alpha}}{\G\V\TB_{\alpha}M:\TW_{\alpha}\tau @A} \andalso
    \infrule[\TTBL]{\G\V M:\TW_{\alpha}\tau @A}{\G\V\TBL_{\alpha}M:\tau @{A\alpha}} \\[2mm]
    \infrule[\TGen]{\G\V M:\tau @A \andalso \alpha\notin\rm{FTV}(\G)\cup\rm{FTV}(A)}{\G\V\Lambda\alpha.M:\forall\alpha.\tau @A} \\[2mm]
    \infrule[\TIns]{\G\V M:\forall\alpha.\tau @A}{\G\V M\ B:\tau[\alpha \mapsto B]@A} \andalso
    \infrule[\TCsp]{\G\V M:\tau @A}{\G\V \%_\alpha M:\tau @{A\alpha}}
    \caption{Typing Rules.}
    \label{fig:typing-rules}
  \end{center}
\end{figure}

The typing rules of \LMD are shown in Figure~\ref{fig:typing-rules}.
The rule \TConst{} means that a constant can appear at any stage.
The rules \TVar, \TAbs, and \TApp{} are almost the same as those in the simply typed
lambda calculus or \LLF.  Additional conditions are that subterms must be
typed at the same stage (\TAbs{} and \TApp); the type
annotation/declaration on a variable has to be a proper type of kind
$*$ (\TAbs) at the stage where it is declared (\TVar{} and \TAbs).


As in standard dependent type systems, \TConv{} allows us to replace the type
of a term with an equivalent one. For example, assuming integers and
arithmetic, a value of type $\textrm{Vector}\ (4+1)$ can also have type
$\textrm{Vector}\ 5$ because of \TConv{}.

The rules \TTB, \TTBL, \TGen, \TIns, and \TCsp{} are constructs for
multi-stage programming. \TTB{} and \TTBL{} are the same as in \LTP, as we
explained in Section \ref{sec:informal-overview}. The rule \TGen{} for stage
abstraction is straightforward. The condition
$\alpha\notin\rm{FTV}(\G)\cup\rm{FTV}(A)$ ensures that the scope of $\alpha$
is in $M$, and avoids capturing variables elsewhere. The
rule \TIns{} is for applications of stages to stage abstractions. The rule
\TCsp{} is for CSP, which means that, if term $M$ is of type $\tau$ at stage
$A$, then $\%_\alpha M$ is of type $\tau$ at stage $A\alpha$. Note that CSP
is also applied to the type \(\tau\) (although it is implicit) in the
conclusion. Thanks to implicit CSP, the typing rule is the same as in \LTP.

\subsubsection{Kind, Type and Term Equivalence.}

Since the syntax of kinds, types, and terms is mutually recursive,
the corresponding notions of equivalence are also mutually recursive.
They are congruences closed under a few axioms for term equivalence.
Thus, the rules for kind and type equivalences are not very interesting, 
except that implicit CSP is allowed.
We show a few representative rules below.

  {\small
    \begin{center}
      \infrule[\textsc{QK-Csp}]{%
        \G\V K \E J @ A
      }{
        \G\V K \E J @ A\alpha
      }
      \hfil
      \infrule[\QTCsp]{
        \G\V \tau \E \sigma :: *@A
      }{
        \G\V \tau \E \sigma :: *@{A\alpha}
      }
      \\[2mm]
      \infrule[\QTApp]{%
        \G\V \tau \E \sigma :: (\Pi x:\rho.K)@A \andalso
        \G\V M \E N : \rho @A
      }{
        \G\V \tau\ M \E \sigma\ N :: K[x \mapsto M]@A
      }
    \end{center}
  }

We show the rules for term equivalence  in
Figure~\ref{fig:term-equivalence-rules}, omitting
straightforward rules for reflexivity, symmetry, transitivity,
and compatibility.
The rules \QBeta, \QTBLTB, and \QLambda{} correspond to
$\beta$-reduction, $\blacklozenge$-reduction, and $\Lambda$-reduction, respectively.

\begin{figure}[tbp]
  \begin{center}
    \infrule[\QBeta]{\G,x:\sigma@A\V M:\tau @A \andalso \G\V N:\sigma@A}{\G\V(\lambda x:\sigma.M)\ N\E M[x\mapsto N] : \tau[x \mapsto N]@A} \\[2mm]
    \infrule[\QLambda]{\G\V (\Lambda\alpha.M) : \forall\alpha.\tau @A}{\G\V (\Lambda\alpha.M)\ \varepsilon \E M[\alpha \mapsto \varepsilon] : \tau[\alpha \mapsto \varepsilon]@A} \\[2mm]
    \infrule[\QTBLTB]{\G\V M \E N : \tau @A}{\G\V \TBL_\alpha(\TB_\alpha M) \E N : \tau @A} \hfil
    \infrule[\QPercent]{\G\V M:\tau @{A\alpha} \andalso \G\V M:\tau @A}{\G\V\%_\alpha M \E M : \tau @{A\alpha}}
    \caption{Term Equivalence Rules.}
    \label{fig:term-equivalence-rules}
  \end{center}
\end{figure}

The only rule that deserves elaboration is the last rule \QPercent.
Intuitively, it means that the CSP operator applied to term $M$ can be
removed if $M$ is also well-typed at the next stage \(A\alpha\).
For example, constants do not depend on the stage (see \TConst) and
so \(\G\V \%_\alpha c \E c : \tau @ A\alpha\) holds but variables
do depend on stages and so this rule does not apply.

\subsubsection{Example.}

We show an example of a dependently typed code generator in a
hypothetical language based on \LTP.  
This language provides definitions by \textbf{let},
recursive functions (represented by \textbf{fix}), \textbf{if}-expressions,
and primitives cons, head, and tail to manipulate vectors. We assume that
$\text{cons}$ is of type $\Pi n:\I.\I \to \text{Vector}\ n \to \text{Vector}\ (n+1)$, 
$\text{head}$ is of type $\Pi n:\I.\text{Vector}\ (n+1) \to \I$, and
$\text{tail}$ is of type $\Pi n:\I.\text{Vector}\ (n+1) \to (\text{Vector}\ n)$.

Let's consider an application, for example, in computer graphics, in which we
have potentially many pairs of vectors of the fixed (but statically unknown)
length and a function---such as vector addition---to be applied to
them. This function should be fast because it is applied many times and be
safe because just one runtime error may ruin the whole long-running calculation.


\newcommand{\Vpn}{\text{Vector}\ (\%_\alpha n)}

Our goal is to define the function vadd of type
\[
  \Pi n:\I.\F\beta.\TW_\beta(\Vpn\to\Vpn\to\Vpn).
\]
\renewcommand{\Vpn}{\text{Vector}\ n}
It takes the length $n$ and returns ($\beta$-closed) code of a
function to add two vectors of length $n$.  The generated
code is run by applying it to \(\varepsilon\) to obtain
a function of type \(\Vpn\to\Vpn\to\Vpn\) as expected.

We start with the helper function vadd$_1$, which takes a stage, the length $n$ of vectors, and two quoted vectors as arguments and returns code that computes the addition of the given two vectors:
\begin{tabbing}
	  $\textbf{let}\ \text{vadd}_1 : \F\alpha.\Pi n:\I.\TW_\alpha\Vpn\to\TW_\alpha\Vpn\to\TW_\alpha\Vpn$                                \\
	  \hspace{6mm} \= $= \textbf{fix}\ f.\Lambda\alpha.\lambda n:\I.\ \lambda v_1:\TW_\alpha\Vpn.\ \lambda v_2:\TW_\alpha\Vpn.$            \\
	  \> \hspace{6mm} \= $\textbf{if}\ n = 0 \ \textbf{then}\ \TB_\alpha \text{nil}$ \\
	  \>\> $\textbf{else}\ \TB_\alpha ($ \= $\textbf{let}\ t_1 = \text{tail}\ (\TBL_\alpha v_1)\ \textbf{in}$ \\
	  \>\>\> $\textbf{let}\ t_2 = \text{tail}\ (\TBL_\alpha v_2)\ \textbf{in}$ \\
          \>\>\> $\text{cons}\ $\= $(\text{head}\ (\TBL_\alpha v_1) + \text{head} \ (\TBL_\alpha v_2))$ \\
          \>\>\>\> $\TBL_\alpha (f\ (n-1)\ (\TB_\alpha t_1)\ (\TB_\alpha t_2)))$
\end{tabbing}
Note that the generated code will not contain branching on $n$ or recursion.
(Here, we assume that the type system can determine whether $n=0$ when
\textbf{then}- and \textbf{else}-branches are typechecked so that
both branches can be given type \(\TW_\alpha \text{Vector }n\).)

Using vadd$_1$, the main function vadd can be defined as follows:
\renewcommand{\Vpn}{\text{Vector}\ (\%_\beta n)}
\begin{align*}
	  & \textbf{let}\ \text{vadd}: \Pi n:\I.\F\beta.\TW_\beta(\Vpn\to\Vpn\to\Vpn)                \\ 
	  & \hspace{6mm} = \lambda n:\I.\Lambda\beta.\TB_\beta (\lambda v_1:\Vpn.\ \lambda v_2:\Vpn. \\
	  & \hspace{63mm} \TBL_\beta (\text{vadd}_1\ \beta\ n\ (\TB_\beta\ v_1)\ (\TB_\beta\ v_2))) 
\end{align*}
\renewcommand{\Vpn}{\text{Vector}\ (\%_\beta 5)}%
The auxiliary function vadd$_1$ generates code to compute addition of
the formal arguments $v_1$ and $v_2$ without branching on $n$ or recursion.
As we mentioned already, if this function is applied to a
(nonnegative) integer constant, say 5, it returns function code for adding
two vectors of size 5.  The type of vadd\ 5, obtained by
substituting 5 for $n$, is
$\F\beta.\TW_\beta(\Vpn\to\Vpn\to\Vpn)$.
\renewcommand{\Vpn}{\text{Vector}\ 5}
If the obtained code is run by applying to \(\varepsilon\),
the type of vadd\ 5\ $\varepsilon$ is
\(\Vpn\to\Vpn\to\Vpn\) as expected.

There are other ways to implement the vector addition function:
by using tuples instead of lists if the length
for all the vectors is statically known or by checking dynamically the
lengths of lists for every pair.  However, our method is better than
these alternatives in two points.  First, our function, $\text{vadd}_1$
can generate functions for vectors of arbitrary length unlike the one
using tuples.  Second, $\text{vadd}_1$ has an advantage in speed over
the one using dynamic checking because it can generate an optimized
function for a given length.

We make two technical remarks before proceeding:
\begin{enumerate}
\item If the generated function code is composed with another piece of code of type, say,
\(\TW_\gamma \text{Vector }5\), \QPercent{} plays an essential role; that is,
\(\text{Vector }5\) and \(\text{Vector }(\%_\gamma 5)\), which would occur
by applying the generated code to \(\gamma\) (instead of \(\varepsilon\)), are syntactically
different types but \QPercent{} enables to equate them.
Interestingly, Hanada and Igarashi~\cite{Hanada2014} rejected the idea of
reduction that removes $\%_\alpha$ when they developed \LTP{}, as such
reduction does not match the operational behavior of the CSP operator in
implementations. However, as an equational system for multi-stage programs,
the rule \QPercent{} makes sense.
\item \renewcommand{\Vpn}{\text{Vector}\ n}
By using implicit type-level CSP, the type of vadd
could have been written
\(\Pi n:\I.\F\beta.\TW_\beta(\Vpn\to\Vpn\to\Vpn)\).  In this type,
Vector\ $n$ is given kind at stage \(\varepsilon\) and type-level CSP
implicitly lifts it to stage \(\beta\).  However, if a type-level
constant takes two or more arguments from different stages, term-level
CSP is necessary.  A matrix type (indexed by the numbers of columns
and rows) would be such an example.
\end{enumerate}

\subsection{Staged Semantics}

The reduction given above is full reduction and any redexes---even
under $\TB_\alpha$---can be reduced in an arbitrary order.
Following previous work~\cite{Hanada2014},
we introduce (small-step, call-by-value) staged semantics,
where only $\beta$-reduction or $\Lambda$-reduction at stage $\varepsilon$ or the outer-most $\blacklozenge$-reduction are allowed,
modeling an implementation.

We start with the definition of values. Since terms under quotations are
not executed, the grammar is indexed by stages.

\begin{definition}[Values]
  The family $V^A$ of sets of values, ranged over by $v^A$,
  is defined by the following grammar.  In the grammar, $A' \neq \varepsilon$ is assumed.
  \begin{align*}
    v^\varepsilon \in V^\varepsilon & ::= \lambda x:\tau.M \mid\ \TB_\alpha v^\alpha \mid \Lambda\alpha.v^\varepsilon                                             & \\
    v^{A'} \in V^{A'}               & ::= x \mid \lambda x:\tau.v^{A'} \mid v^{A'}\ v^{A'} \mid\ \TB_\alpha v^{A'\alpha} \mid \Lambda\alpha.v^{A'} \mid v^{A'}\ B & \\
                                    & \quad\   \mid\ \TBL_\alpha v^{A''} (\text{if } A' = A''\alpha \text{ for some } \alpha, A'' \neq \varepsilon)               & \\
                                    & \quad\   \mid\ \%_\alpha v^{A''} (\text{if } A' = A''\alpha)
  \end{align*}
\end{definition}

Values at stage $\varepsilon$ are $\lambda$-abstractions, quoted pieces of code,
or $\Lambda$-abstractions.  The body of a $\lambda$-abstraction can
be any term but the body of $\Lambda$-abstraction has to be a value.  It
means that the body of $\Lambda$-abstraction must be evaluated.  The
side condition for $\TBL_\alpha v^{A'}$ means that escapes in a value
can appear only under nested quotations because an escape under a
single quotation will splice the code value into the surrounding
code.  See Hanada and Igarashi~\cite{Hanada2014} for details.

In order to define staged reduction, we define redex and evaluation contexts.

\begin{definition}[Redex]
  The sets of $\varepsilon$-redexes (ranged over by $R^\varepsilon$) and $\alpha$-redexes (ranged over by $R^\alpha$) are defined by the following grammar.
  \begin{align*}
     & R^\varepsilon ::= (\lambda x:\tau.M)\ v^\varepsilon \mid (\Lambda\alpha.v^\varepsilon)\ \varepsilon \\
     & R^\alpha      ::=\ \TBL_\alpha \TB_\alpha v^\alpha                                                         \\
  \end{align*}
\end{definition}

\begin{definition}[Evaluation Context]
  Let $B$ be either \(\varepsilon\) or a stage variable \(\beta\).
  The family of sets $ECtx^A_B$ of evaluation contexts, ranged over by $E^A_B$, is defined by the following grammar (in which $A'$ stands for a non-empty stage).
  \begin{align*}
    E^\varepsilon_B \in ECtx^\varepsilon_B & ::= \square\ (\text{if\ } B = \varepsilon)
    \mid E^\varepsilon_B\ M \mid v^\varepsilon\ E^\varepsilon_B \mid \TB_\alpha E^\alpha_B
    \mid \Lambda\alpha.E^\varepsilon_B \mid E^\varepsilon_B\ A                                                                                    \\
    E^{A'}_B \in ECtx^{A'}_B               & ::= \square\ (\text{if } A' = B) \mid \lambda x:\tau.E^{A'}_B \mid E^{A'}_B\ M \mid v^{A'}\ E^{A'}_B \\
                                           & \mid \TB_\alpha E^{A'\alpha}_B \mid \TBL_\alpha E^{A}_B \ (\text{where } A\alpha = A')               \\
                                           & \mid \Lambda\alpha.E^{A'}_B \mid E^{A'}_B\ A \mid \%_\alpha\ E^{A}_B \ (\text{where } A\alpha = A')
  \end{align*}
\end{definition}

The subscripts $A$ and $B$ in $E^A_B$ stand for the stage of the evaluation
context and of the hole, respectively. The grammar represents that staged
reduction is left-to-right and call-by-value and terms under \(\Lambda\) are
reduced. Terms at non-$\varepsilon$ stages are not reduced, except
redexes of the form \(\TBL_\alpha \TB_\alpha v^\alpha\) at stage \(\alpha\).
A few examples of
evaluation contexts are shown below:
\begin{align*}
  \square\ (\lambda x:\I.x)                  & \in  ECtx^\varepsilon_\varepsilon \\
  \Lambda\alpha.\square\ \varepsilon            & \in ECtx^\varepsilon_\varepsilon  \\
  \TBL_\alpha \TB_\alpha \TBL_\alpha \square & \in ECtx^\alpha_\varepsilon
\end{align*}
We write $E^A_B[M]$ for the term obtained by filling the hole $\square$ in $E^A_B$ by $M$.

Now we define staged reduction using the redex and evaluation contexts.

\begin{definition}[Staged Reduction]\sloppy
  The staged reduction relation, written $M \longrightarrow_s N$, is defined by
  the least relation closed under the rules below.
  \begin{align*}
    E^A_\varepsilon [(\lambda x:\tau.M)\ v^\varepsilon] & \longrightarrow_s E^A_\varepsilon[M[x\mapsto v^\varepsilon]]      \\
    E^A_\varepsilon [(\Lambda\alpha.v^\varepsilon)\ A]  & \longrightarrow_s E^A_\varepsilon[v^\varepsilon[\alpha\mapsto A]] \\
    E^A_\alpha [\TBL_\alpha \TB_\alpha v^\alpha]        & \longrightarrow_s E^A_\alpha[v^\alpha]                            \\
  \end{align*}
\end{definition}

This reduction relation reduces a term in a deterministic,
left-to-right, call-by-value manner.  An application of an abstraction
is executed only at stage \(\varepsilon\) and only a quotation at
stage \(\varepsilon\) is spliced into the surrounding code---notice
that, if \(\TB_\alpha v^\alpha\) is at stage \(\varepsilon\), then the
redex \(\TBL_\alpha \TB_\alpha v^\alpha\) is at stage \(\alpha\).
In other words, terms in brackets are not evaluated until the terms are run
and arguments of a function are evaluated before the application.
We show an example of staged reduction.
Underlines show the redexes.
\begin{align*}
                    & (\Lambda\alpha.(\TB_\alpha \underline{\TBL_\alpha \TB_\alpha ((\lambda x:\I.x)\ 10))})\ \varepsilon \\
  \longrightarrow_s & \underline{(\Lambda\alpha.(\TB_\alpha ((\lambda x:\I.x)\ 10)))\ \varepsilon}                        \\
  \longrightarrow_s & \underline{(\lambda x:\I.x)\ 10}                                                                    \\
  \longrightarrow_s & 10
\end{align*}


\section{Properties of \LMD \label{sec:properties}}

In this section, we show the basic properties of \LMD: preservation, strong
normalization, confluence for full reduction, and progress for staged
reduction.


The Substitution Lemma in \LMD{} is a little more complicated than usual
because there are eight judgment forms and two kinds of substitution.
The Term Substitution Lemma states that term substitution $[z \mapsto M]$
preserves derivability of judgments. The Stage Substitution Lemma
states similarly for stage substitution $[\alpha\mapsto A]$.

We let $\mathcal{J}$ stand for the judgments $K \iskind @A$, $\tau::K@A$,
$M:\tau@A$, $K \E J @ A$, $\tau \E \sigma @ A$, and
$M \E N : \tau @ A$.  Substitutions $\mathcal{J}[z \mapsto M]$ and
$\mathcal{J}[\alpha \mapsto A]$ are defined in a straightforward
manner.  Using these notations, the two substitution lemmas are stated as follows:

We proved the next two leammas by simultaneous induction on derivations.

\begin{lemma}[Term Substitution]
  If $\G, z:\xi@B, \D \V \mathcal{J}$ and $\G\V N:\xi @B$, then $\G, (\D[z \mapsto N]) \V \mathcal{J}[z \mapsto N]$.  Similarly, if $\V \G, z:\xi@B, \D$ and
  $\G\V N:\xi @B$, then $\V \G, (\D[z \mapsto N])$.
\end{lemma}

\begin{lemma}[Stage Substitution]
	If $\G \V \mathcal{J}$, then $\G[\beta\mapsto B] \V \mathcal{J}[\beta\mapsto B]$.  Similarly, if $\V \G$, then $\V \G[\beta\mapsto B]$.
\end{lemma}


The following Inversion Lemma is needed to prove the main theorems.
As usual~\cite{TAPL}, the Inversion Lemma enables us to infer the types of subterms of a term from the shape of the term.

\begin{lemma}[Inversion]\ 
	\begin{enumerate}
		\item If $\G \V (\lambda x:\sigma.M) : \rho$ then there are $\sigma'$ and $\tau'$ such that
		      $\rho = \Pi x:\sigma'.\tau'$, $\G \V \sigma \E \sigma'@A$ and $\G ,x:\sigma'@A\V M:\tau'@A$.
		\item If $\G \V \TB_\alpha M : \tau@A$ then 
		      there is $\sigma$ such that $\tau = \TW_\alpha \sigma$ and $\G \V M : \sigma@A$.
		  \item If $\G \V \Lambda\alpha.M : \tau$ then 
		  there is $\sigma$ such that $\sigma = \forall\alpha.\sigma$ and $\G \V M : \sigma@A$.
	\end{enumerate}
\end{lemma}

\begin{proof}
  Each item is strengthened by statements about type equivalence.
  For example, the first statement is augmented by
  \begin{quotation}
    If $\G \V \rho \E (\Pi x:\sigma.\tau) : K @A$, then there exist
    $\sigma'$ and $\tau'$ such that $\rho = \Pi x:\sigma'.\tau'$ and
    $\G \V \sigma \E \sigma' : K @A$ and
    $\G, x:\sigma@A\V \tau \E \tau' : J @A$.
  \end{quotation}
  and its symmetric version.  Then, they are proved simultaneously by induction on derivations.
  Similarly for the others.
  \qed
\end{proof}


Thanks to Term/Stage Substitution and Inversion, we can prove Preservation easily.

\begin{theorem}[Preservation]
	If $\G\V M:\tau @A$ and $M \longrightarrow M'$, then $\G\V M':\tau @A$.
\end{theorem}

\begin{proof}
	First, there are three cases for $M \longrightarrow M'$.
	They are $M \longrightarrow_\beta M'$, $M \longrightarrow_\Lambda M'$, and $M \longrightarrow_\blacklozenge M'$.
	For each case, we can use straightforward induction on typing derivations.
        \qed
\end{proof}


Strong Normalization is also an important property, which guarantees that
no typed term has an infinite reduction sequence.
Following standard proofs (see, e.g., \cite{harper1993framework}), we prove this theorem by translating \LMD to the simply typed lambda calculus.

\begin{theorem}[Strong Normalization]
	If $\G\V M_1:\tau@A$ then there is no infinite sequence $(M_i)_{i\ge1}$ of terms such that
	$M_i \longrightarrow M_{i+1}$ for $i\ge 1$.
\end{theorem}

\begin{proof}
	In order to prove this theorem, we define a translation $(\cdot)^\natural$ from \LMD\ to the simply typed lambda calculus.
	Second, we prove the $\natural$-translation preserves typing and reduction.
	Then, we can prove Strong Normalization of \LMD from Strong Normalization of the simply typed lambda calculus.
 \qed
\end{proof}

Confluence is a property that any reduction sequences from one typed term converge.
Since we have proved Strong Normalization, we can use Newman's Lemma~\cite{DBLP:books/daglib/0092409} to prove Confluence.

\begin{theorem}[Confluence]
	For any term $M$, if $M \longrightarrow^* M'$ and $M \longrightarrow^* M''$ then
	there exists $M'''$ that satisfies $M' \longrightarrow^* M'''$ and $M'' \longrightarrow^* M'''$.
\end{theorem}

\begin{proof}
  We can easily show Weak Church-Rosser.  Use Newman's Lemma.
 \qed
\end{proof}

Now, we turn our attention to staged semantics.  First, the staged
reduction relation is a subrelation of full reduction, so Subject
Reduction holds also for the staged reduction.

\begin{theorem}
  If $M \longrightarrow_s M'$, then $M \longrightarrow M'$.
\end{theorem}
\begin{proof}
  Easy.
 \qed
\end{proof}

The following theorem Unique Decomposition ensures that every typed
term is either a value or can be uniquely decomposed to an evaluation
context and a redex, ensuring that a well-typed term is not
immediately stuck and the staged semantics is deterministic.

\begin{theorem}[Unique Decomposition]
  If $\G$ does not have any variable declared at stage $\varepsilon$ 
  and $\G \V M : \tau @ A$, then either
  \begin{enumerate}
  \item $ M \in V^A$, or
  \item $M$ can be uniquely decomposed into an evaluation context and a redex, that is, there uniquely exist $B, E^A_B$, and $R^B$ such that $M = E^A_B[R^B]$.
  \end{enumerate}
\end{theorem}

\begin{proof}
  We can prove by straightforward induction on typing derivations.
  \qed
\end{proof}

The type environment $\G$ in the statement usually has to be empty;
in other words, the term has to be closed.  The condition is relaxed here
because variables at stages higher than \(\varepsilon\) are considered
symbols.  In fact, this relaxation is required for proof by induction
to work.

      
Progress is a corollary of Unique Decomposition.

\begin{corollary}[Progress]
	If $\G$ does not have any variable declared at stage $\varepsilon$ and $\G \V M : \tau  @ A$, then
	$ M \in V^A $ or there exists $M'$ such that $M \longrightarrow_s M'$.
\end{corollary}


\section{Related Work}
\label{sec:related-work}


MetaOCaml is a programming language with quoting, unquoting, run, and CSP.
Kiselyov~\cite{8384206} describes many applications of MetaOCaml, including
filtering in signal processing, matrix-vector product, and a DSL compiler.


Theoretical studies on multi-stage programming owe a lot to seminal work by
Davies and Pfenning~\cite{DaviesPfenning01JACM} and
Davies~\cite{davies1996temporal}, who found Curry-Howard correspondence
between multi-stage calculi and modal logic. In particular, Davies'
$\lambda^\circ$~\cite{davies1996temporal} has been a basis for several
multi-stage calculi with quasi-quotation. $\lambda^\circ$ did not have
operators for run and CSP; a few
studies~\cite{benaissa1999logical,MoggiTahaBenaissaSheard99ESOP} enhanced and
improved $\lambda^\circ$ towards the development of a type-safe multi-stage
calculus with quasi-quotation, run, and CSP, which were proposed by Taha and
Sheard as constructs for multi-stage programming~\cite{MetaML}. 
Finally, Taha and Nielsen invented the concept of environment
classifiers~\cite{taha2003environment} and developed a typed calculus
$\lambda^\alpha$, which was equipped with all the features above in a type
sound manner and formed a basis of earlier versions of MetaOCaml. Different
approaches to type-safe multi-stage programming with slightly different
constructs for composing and running code values have been studied by Kim,
Yi, and Calcagno~\cite{DBLP:conf/popl/KimYC06} and Nanevski and
Pfenning~\cite{DBLP:journals/jfp/NanevskiP05}.

Later, Tsukada and Igarashi~\cite{Tsukada} found correspondence
between a variant of \(\lambda^\alpha\) called $\lambda^\TW$
and modal logic and showed that run could be represented as a special
case of application of a transition abstraction ($\Lambda\alpha.M$) to
the empty sequence $\varepsilon$.  Hanada and
Igarashi~\cite{Hanada2014} developed \LTP as an extension
$\lambda^\TW$ with CSP.


There is much work on dependent types and most of it is affected by
the pioneering work by Martin-L\"{o}f~\cite{martin1973intuitionistic}.
Among many dependent type systems such as
$\lambda^\Pi$~\cite{Meyer1986}, The Calculus of
Constructions~\cite{coquand:inria-00076024}, and Edinburgh
LF~\cite{harper1993framework}, we base our work on \LLF~\cite{attapl}
(which is quite close to $\lambda^\Pi$ and Edinburgh LF) due to its
simplicity.  It is well known that dependent types are useful to
express detailed properties of data structures at the type level such
as the size of data structures~\cite{Xi98} and typed abstract syntax
trees~\cite{DBLP:conf/dsl/LeijenM99,DBLP:conf/popl/XiCC03}.  The
vector addition discussed in Section~\ref{sec:formal} is also such an
example.






The use of dependent types for code generation is studied by
Chlipala~\cite{chlipala2010ur} and Ebner et
al.~\cite{DBLP:journals/pacmpl/EbnerURAM17}.  They use inductive types
to guarantee well-formedness of generated code.  Aside from the lack
of quasi-quotation, their systems are for heterogeneous
meta-programming and compile-time code generation and they do not
support features for run-time code generation such as run and CSP, as \LMD{} does.


We discuss earlier attempts at incorporating dependent types into
multi-stage programming.  Pasalic and Taha~\cite{pasalic2002tagless}
designed \(\lambda_{H\circ}\) by introducing the concept of stage into
an existing dependent type system
\(\lambda_H\)~\cite{zhong2002certified}.  However,
\(\lambda_{H\circ}\) is equipped with neither run nor CSP.  Forgarty,
Pasalic, Siek and Taha~\cite{fogarty2007concoqtion} extended the type
system of MetaOCaml with indexed types.  With this extension, types
can be indexed with a Coq term.  Chen and Xi~\cite{chen2003meta}
introduced code types augmented with information on types of free
variables in code values in order to prevent code with free variables
from being evaluated.  These systems separate the language of type
indices from the term language.  As a result, they do not enjoy
full-spectrum dependent types but are technically simpler because
there is no need to take stage of types into account.  Brady and
Hammond~\cite{brady2006dependently} have discussed a combination of
(full-spectrum) dependently typed programming with staging in the
style of MetaOCaml to implement a staged interpreter, which is
statically guaranteed to generate well-typed code.  However, they
focused on concrete programming examples and there is no theoretical
investigation of the programming language they used.

Berger and Tratt~\cite{martin2015HGRTMP} gave program logic for
\(\text{Mini-ML}^\square_e\)~\cite{DaviesPfenning01JACM}, which would
allow fine-grained reasoning about the behavior of code generators.
However, it cannot manipulate open code which ours can deal with.




\section{Conclusion \label{sec:conclusion}}

We have proposed a new multi-stage calculus \LMD with dependent types,
which make it possible for programmers to express finer-grained
properties about the behavior of code values.  The combination leads
to augmentation of almost all judgments in the type system with stage
information.  CSP and type equivalence (specially tailored for CSP) are
keys to expressing dependently typed practical code generators.  We
have proved basic properties of \LMD, including preservation,
confluence, strong normalization for full reduction, and progress for
staged reduction.

Developing a typechecking algorithm for \LMD is left for future
work.  We expect that most of the development is straightforward,
except for implicit CSP at the type-level and \%-erasing equivalence
rules.

\subsection*{Acknowledgments.}
We would like to thank John Toman, Yuki Nishida, and anonymous reviewers for useful comments.  We would also like to thank Ayato Yokoyama for pointing out an error in the type system.

\bibliographystyle{splncs04}
\bibliography{main}
%
%
%
%
%

\iffullversion
\appendix

\newtheorem{dfn}{Definition}
\newtheorem{ex}{Example}
\newtheorem{cm}{Comment}
\newcommand{\figheader}[2]{
  \begin{flushleft}
    #2 {\bf \normalsize #1}
\end{flushleft}}

\newpage
\section{Full Definition of \LMD}
\subsection{Syntax}

\begin{align*}
	\textrm{Terms}             &  & M,N,L,O,P                & ::= c \mid x \mid \lambda x:\tau.M\ \mid M\ M \mid \TB_\alpha M                                                                                    \\
	                           &  &                          & \ \ \ \ \mid \TBL_\alpha M \mid \Lambda\alpha.M \mid M\ A \mid M\ \varepsilon \mid \%_\alpha M                                                                       \\
	\textrm{Types}             &  & \tau,\sigma,\rho,\pi,\xi & ::= X \mid \Pi x:\tau.\tau \mid \tau\ M \mid \TW_{\alpha} M \mid \F\alpha.\tau                                                                     \\
	\textrm{Kinds}             &  & K,J,I,H,G                & ::= * \mid \Pi x:\tau.K                                                                                                                            \\
	\textrm{Type environments} &  & \Gamma                   & ::= \emptyset \mid \Gamma,x:\tau@A\\
	\textrm{Signature}         &  & \Sigma                   & ::= \emptyset \mid \Sigma,X::K\\
	                           &  &                          & \ \ \ \ \mid \Sigma,c:\tau\\
	\textrm{Stage variables}   &  &                          & \alpha,\beta,\gamma,...                                                                                                                            \\
	\textrm{Stage}             &  &                          & A,B,C,...                                                                                                                                          \\
	\textrm{Variables}         &  &                          & x,y,z,...                                                                                                                                          \\
	\textrm{Type variables}    &  &                          & X,Y,Z,...                                                                                                                                          \\
\end{align*}

\subsection{Reduction}

\figheader{Term reduction}{\rulefbox{M \longrightarrow N}}
\begin{center}
	\begin{align*}
		 & (\lambda x:\tau.M) N \longrightarrow_\beta M[x \mapsto N]                             \\
		 & \TBL_\alpha (\TB_\alpha M)\longrightarrow_\blacklozenge M                             \\
		 & (\Lambda \alpha.M)\ A \longrightarrow_\Lambda M[\alpha \mapsto A]
	\end{align*}
\end{center}

\subsection{Type System}

\figheader{Well-formed signatures}{\rulefbox{\vdash\Sigma}}
{\small
\begin{center}
  \infrule{
  }{
    \vdash \emptyset
  }
  \hfil
  \infrule{
    \vdash \Sigma \andalso
    \V K \iskind @ \varepsilon \andalso
    X\notin\textit{dom}(\Sigma)
  }{
    \vdash \Sigma, X::K
  }
  \\[2mm]
  \infrule{
    \vdash \Sigma \andalso
    \V \tau :: * @ \varepsilon \andalso
    c\notin\textit{dom}(\Sigma)
  }{
   \vdash \Sigma, c:\tau
 }
\end{center}
}
\figheader{Well-formed type environments}{\rulefbox{\V \Gamma}}

{\small
\begin{center}
  \infrule{
  }{
    \V \emptyset
  }
  \hfil
  \infrule{
    \V \Gamma \andalso
    \Gamma \V \tau :: * @ A \andalso
    x\notin\textit{dom}(\Sigma)
  }{
   \V \Gamma, x:\tau@A
 }
\end{center}
}

\figheader{Well-formed kinds}{\rulefbox{\Gamma \vdash K\iskind @A}}
{\small
\begin{center}
  \infrule[\WStar]{
  }{
    \G\V*\iskind @A
  } \hfil
  \infrule[\WAbs]{
    \G\V \tau::*@A \andalso \G,x:\tau@A\V K\iskind @A
  }{
    \G\V(\Pi x:\tau.K)\iskind @A
  }
\end{center}
}

\figheader{Kinding}{\rulefbox{\G \V \tau::K @ A}}
{\small
\begin{center}
  \infrule[{\KTConst}]{
    X::K \in \Sigma 
  }{
    \G \V X::K@A
  } \\[2mm]
  \infrule[{\KAbs}]{
    \G\V \tau :: *@A \andalso \G,x:\tau@A\V \sigma::*@A
  }{
    \G\V(\Pi x:\tau.\sigma) :: *@A
  } \\[2mm]
  \infrule[{\KApp}]{
    \G\V \sigma:: (\Pi x:\tau.K)@A \andalso \G\V M:\tau@A
  }{
    \G\V \sigma\ M::K[x\mapsto M]@A
  } \\[2mm]
  \infrule[{\KConv}]{
    \G\V \tau::K@A \andalso
    \G\V K\equiv J@A
  }{
    \G\V \tau::J@A
  } \hfil
  \infrule[{\KTW}]{
    \G\V \tau::*@A\alpha
  }{
    \G\V\TW_\alpha \tau::*@A
  }\\[2mm]
  \infrule[{\KGen}]{
    \G\V \tau::K@A \andalso
    \alpha\notin\rm{FTV}(\G)\cup\rm{FTV}(A)
  }{
    \G\V\forall\alpha.\tau::K@A
  } \hfil
  \infrule[{\KCsp}]{
    \G\V \tau::*@A
  }{
    \G\V \tau::*@A\alpha
  }
\end{center}
}

\figheader{Typing}{\rulefbox{\G\V M:\tau @A}}
{\small
\begin{center}
  \infrule[{\TConst}]{
    c:\tau \in \Sigma \andalso
  }{
    \G \V c:\tau@A
  }\andalso
  \infrule[{\TVar}]{
    x:\tau@A \in \G \andalso
  }{
    \G \V x:\tau@A
  } \\[2mm]
  \infrule[\TAbs]{
    \G\V \sigma::*@A \andalso
    \G,x:\sigma@A\V M:\tau@A
  }{
    \G\V(\lambda (x:\sigma).M):(\Pi (x:\sigma).\tau)@A
  } \\[2mm]
  \infrule[{\textsc{T-App}}]{
    \G\V M:(\Pi (x:\sigma).\tau)@A \andalso
    \G\V N:\sigma@A
  }{
    \G\V M\ N : \tau[x\mapsto N]@A
  } \\[2mm]
  \infrule[{\textsc{T-Conv}}]{
    \G\V M:\tau@A \andalso
    \G\V \tau\equiv \sigma :: K@A
  }{
    \G\V M:\sigma@A
  } \\[2mm]
  \infrule[{\TTB}]{
    \G\V M:\tau@{A\alpha}
  }{
    \G\V\TB_{\alpha}M:\TW_{\alpha}\tau@A
  } \andalso
  \infrule[{\TTBL}]{
    \G\V M:\TW_{\alpha}\tau@A
  }{
    \G\V\TBL_{\alpha}M:\tau@{A\alpha}
  } \\[2mm]
  \infrule[\TGen]{
    \G\V M:\tau@A \andalso
    \alpha\notin\rm{FTV}(\G)\cup\rm{FTV}(A)
  }{
    \G\V\Lambda\alpha.M:\forall\alpha.\tau@A
  } \\[2mm]
  \infrule[\TIns]{
    \G\V M:\forall\alpha.\tau@A
  }{
    \G\V M\ A:\tau[\alpha \mapsto A]@A
  } \hfil
  \infrule[\TCsp]{
    \G\V M:\tau@A
  }{
    \G\V \%_\alpha M:\tau@{A\alpha}
  }
\end{center}
}

\figheader{Kind Equivalence}{\rulefbox{\G\V K\E J@A}}
{\small
\begin{center}
  \infrule[{\QKAbs}]{
    \G\V \tau \E \sigma :: *@A \andalso
    \G,x:\tau@A \V K \E J@A
  }{
    \G\V\Pi x:\tau.K \E \Pi x:\sigma.J@A
  } \\[2mm]
  \infrule[{\textsc{QK-Csp}}]{
    \G\V K \E J@A
  }{
    \G\V K \E J@{A\alpha}
  } \hfil
  \infrule[\QKRefl]{
    \G\V K \iskind @A
  }{
    \G\V K\E K@A
  } \\[2mm]
  \infrule[\QKSym]{
    \G\V K \E J@A
  }{
    \G\V J \E K@A
  } \hfil
  \infrule[\QKTrans]{
    \G\V K \E J@A \andalso
    \G\V J \E I@A
  }{
    \G\V K \E I@A
  }
\end{center}
}
\figheader{Type Equivalence}{\rulefbox{\G\V S\E T :: K @A}}
{\small
\begin{center}
  \infrule[{\QTAbs}]{
    \G\V \tau \E \sigma :: *@A \andalso
    \G,x:\tau@A \V \rho \E \pi :: *@A
  }{
    \G\V\Pi x:\tau.\rho \E \Pi x:\sigma.\pi :: *@A
  } \\[2mm]
  \infrule[{\QTApp}]{
    \G\V \tau \E \sigma :: (\Pi x:\rho.K)@A \andalso
    \G\V M \E N : \rho @A
  }{
    \G\V \tau\ M \E \sigma\ N :: K[x \mapsto M]@A
  } \\[2mm]
  \infrule[\QTTW]{
    \G\V \tau \E \sigma :: *@{A\alpha}
  }{
    \G\V \TW_{\alpha} \tau \E \TW_{\alpha} \sigma :: *@A
  } \\[2mm]
  \infrule[\QTGen]{
    \G\V \tau \E \sigma :: *@A \andalso
    \alpha\notin\rm{FTV}(\G)\cup\rm{FTV}(A)
  }{
    \G\V \forall\alpha.\tau \E  \forall\alpha.\sigma :: *@A
  } \\[2mm]
  \infrule[\QTCsp]{
    \G\V \tau \E \sigma :: *@A
  }{
    \G\V \tau \E \sigma :: *@{A\alpha}
  }
  \infrule[\QTRefl]{
    \G\V \tau::K@A
  }{
    \G\V \tau\E\tau :: K@A
  } \hfil
  \infrule[\QTSym]{
    \G\V \tau \E \sigma :: K@A
  }{
    \G\V \sigma \E \tau :: K@A
  } \\[2mm]
  \infrule[\QTTrans]{
    \G\V \tau \E \sigma :: K@A \andalso
    \G\V \sigma \E \rho  :: K@A
  }{
    \G\V \tau \E \rho  :: K@A
  }
\end{center}
}

\figheader{Term Equivalence}{\rulefbox{\G\V M\E N : \tau @A}}
{\small
\begin{center}
  \infrule[\QAbs]{
    \G\V \tau \E \sigma :: *@A \andalso
    \G,x:\tau@A \V M \E N : \rho @A
  }{
    \G\V\lambda x:\tau.M \E \lambda x:\sigma.N : (\Pi x:\tau.\rho)@A
  } \\[2mm]
  \infrule[\QApp]{
    \G\V M \E L : (\Pi x:\sigma.\tau)@A \andalso
    \G\V N \E O : \sigma@A
  }{
    \G\V M\ N \E L\ O : \tau[x \mapsto N]@A
  } \\[2mm]
  \infrule[\QTB]{
    \G\V M \E N : \tau@{A\alpha}
  }{
    \G\V \TB_\alpha M \E \TB_\alpha N : \TW_\alpha \tau@A
  } \hfil
  \infrule[{\QTBL}]{
    \G\V M \E N : \TW_\alpha \tau@A
  }{
    \G\V \TBL_\alpha M \E \TBL_\alpha N : \tau@{A\alpha}
  } \\[2mm]
  \infrule[{\QGen}]{
    \G\V M\E N : \tau@A \andalso
    \alpha \notin \FTV(\G)\cup\FTV(A)
  }{
    \G\V \Lambda\alpha.M \E \Lambda\alpha.N : \forall\alpha.\tau@A
  } \\[2mm]
  \infrule[{\QIns}]{
    \G\V M \E N:\forall\alpha.\tau@A
  }{
    \G\V M\ A \E N\ A : \tau[\alpha \mapsto A]@A
  }\hfil
  \infrule[\QCsp]{
    \G\V M \E N : \tau @A
  }{
    \G\V\%_\alpha M \E \%_\alpha N : \tau@{A\alpha}
  } \\[2mm]
  \infrule[\QRefl]{
    \G\V M:\tau@A
  }{
    \G\V M\E M : \tau@A
  } \hfil
  \infrule[\QSym]{
    \G\V M\E N : \tau@A
  }{
    \G\V N\E M : \tau@A
  } \\[2mm]
  \infrule[\QTrans]{
    \G\V M\E N : \tau@A \andalso
    \G\V N\E L : \tau@A
  }{
    \G\V M\E L : \tau@A
  } \\[2mm]
  \infrule[{\QBeta}]{
    \G,x:\sigma@A\V M:\tau@A \andalso
    \G\V N:\sigma@A
  }{
    \G\V(\lambda x:\sigma.M)\ N\E M[x\mapsto N] : \tau[x \mapsto N]@A
  } \\[2mm]
  \infrule[{\QTBLTB}]{
    \G\V M \E N : \tau@A
  }{
    \G\V \TBL_\alpha(\TB_\alpha M) \E N : \tau @A
  } \\[2mm]
  \infrule[{\QLambda}]{
    \G\V (\Lambda\alpha.M) : \forall\alpha.\tau@A
  }{
    \G\V (\Lambda\alpha.M)\ B \E M[\alpha \mapsto B] : \tau[\alpha \mapsto B]@A
  } \\[2mm]
  \infrule[{\QPercent}]{
    \G\V M:\tau@{A\alpha} \andalso
    \G\V M:\tau@A
  }{
    \G\V\%_\alpha M \E M : \tau@{A\alpha}
  }
\end{center}
}



\section{Proofs}
$\mathcal{J}$ is a metavariable for judgments as in Section \ref{sec:properties}.
We say type environment \(\G\) is a subsequence of type environment \(\D\)
if and only if we can get \(\G\) from \(\D\) by deleting some or no variables without changing the order of the remaining elements.
\begin{lemma}[Weakening]
	If \(\G \V \mathcal{J}@A\) and \(\G\) is a subsequence of \(\D\), then \(\D \V J@A\).
\end{lemma}

\begin{proof}
	By straightforward induction on the derivation of typing, kinding, well-formed kinding,
	term equivalence, type equivalence or kind equivalence.
	We show only representative cases.

	\begin{rneqncase}{\WAbs{}}{
			\mathcal{J} = (\Pi x:\tau.K)\iskind @A \\
			\G \V \tau :: * @ A & \G,x:\tau@A \V K \iskind @A
		}
		By the induction hypothesis, we have
		\begin{align*}
			 & \D \V \tau::*@A &  & \D, x:\tau@A \V K \iskind @ A
		\end{align*}
		from which $\D\V (\Pi x:\tau.K) \iskind @ A$ follows by $\WAbs$.
	\end{rneqncase}

	\begin{rneqncase}{\KAbs{}}{
			\mathcal{J} = (\Pi x:\tau.\sigma) :: *@A \\
			\G \V \tau::*@A &  & \G,x:\tau@A\V \sigma :: * @ A.
		}
		By the induction hypothesis, we have
		\begin{align*}
			 & \D \V \tau::*@A &  & \D,x:\tau@A\V \sigma :: * @ A
		\end{align*}
		from which $\D \V (\Pi x:\tau.\sigma) :: *@A$ follows.
	\end{rneqncase}

	\begin{rneqncase}{
			\QTAbs{}}{
			\mathcal{J} = \Pi x:\tau.\rho \E \Pi x:\sigma.\pi :: * @ A \\
			\G\V \tau \E \sigma :: *@A  & \G,x:\tau@A \V \rho \E \pi :: *@A.
		}
		By the induction hypothesis, we have
		\begin{align*}
			 & \D\V \tau \E \sigma :: *@A &  & \D,x:\tau@A \V \rho \E \pi :: *@A
		\end{align*}
		from which \( \D\V  \Pi x:\tau.\rho \E \Pi x:\sigma.\pi :: * @ A \) follows.
	\end{rneqncase}
	\qed\end{proof}

\begin{theorem}[Term Substitution]
	If $\G, z:\xi@B, \D \V \mathcal{J}$ and $\G\V P:\xi @B$, then $\G, \D[z \mapsto P] \V \mathcal{J}[z \mapsto P]$.  Similarly, if $\V \G, z:\xi@B, \D$ and
	$\G\V P:\xi @B$, then $\V \G, \D[z \mapsto P]$.
\end{theorem}
\begin{proof}
	Proved simultaneously by induction on derivations with case analysis on the last rule used.
	We show only representative cases.
		{
			\newcommand{\SB}{[z \mapsto P]}
			\newcommand{\GG}{\G}
			\newcommand{\GGV}{\G \V}

			\begin{rneqncase}{\TVar{}}{
					\mathcal{J} = y:\tau@A\\
					y:\tau@A \in (\G, z:\xi@B, \D).
				}
				\begin{itemize}
					\item If $y:\tau@A \in \G$ or $y:\tau@A \in \D$, then it is obvious that $y:\tau\SB@A \in \GG$.
					      \(\G,\D\SB \V y\SB:\tau\SB@A\) from \TVar.

					\item If $y:\tau@A = z:\xi@B$, then
					      $y = z$, $\tau = \xi$, and $A = B$.
					      From the well-formedness of \( \G, z:\xi@B,\D \), there is no $z$ in $\xi$.
					      Therefore, $\tau\SB = \xi\SB = \xi$.
					      It is obvious that $y\SB = z\SB = P$.
					      Thus, $\G \V y\SB : \tau\SB@A$.
					      By Weakening, $\G,\D\SB \V y\SB : \tau\SB@A$

				\end{itemize}
			\end{rneqncase}

			\begin{rneqncase}{\WAbs{}}{
					\mathcal{J} =  (\Pi x:\tau.K) \iskind @A\\
					\G, z:\xi@B, \D \V \tau::*@A & \G, z:\xi@B, \D, x:\tau \V K \iskind@A.
				}
				We can assume $x \neq z$ because we can select fresh $x$ when we construct $\Pi$ type.
				By the induction hypothesis,
				\begin{align*}
					 & \G,\D\SB \V \tau\SB::*@A &  & \G,\D\SB, x:\tau\SB \V K\SB \iskind@.
				\end{align*}
				$\G,\D\SB \V (\Pi x:\tau\SB.K\SB) \iskind @A$ from \WAbs{} and
				it is equivalent to $\G,\D\SB \V (\Pi x:\tau.K)\SB \iskind @A$.
			\end{rneqncase}

			\begin{rneqncase}{\TApp{}}{
					\mathcal{J} = M\ N:\tau[x\mapsto N]@A\\
					\G, z:\xi@B, \D \V M:\Pi(x:\sigma).\tau@A & \G, z:\xi@B, \D \V N:\sigma@A
				}
				We can assume $x \neq z$ because we can select fresh $x$ when we construct $\Pi$ type.
				By the induction hypothesis,
				\begin{align*}
					 & \G, \D\SB \V M\SB: (\Pi(x:\sigma).\tau)\SB@A \\
					 & \G,\D\SB\V N\SB: \sigma\SB@A.
				\end{align*}
				and \(\G, \D\SB \V M\SB: (\Pi(x:\sigma).\tau)\SB@A\) is equivalent to \( \G, \D\SB \V M\SB: (\Pi(x:\sigma\SB).\tau\SB)@A \).
				From \TApp, \(\G,\D\SB (M\SB\ N\SB): \tau\SB[x \mapsto N\SB]@A\) and this is equivalent to
				\(\G,\D\SB \V (M\ N)\SB: \tau[x \mapsto N]\SB@A\).
			\end{rneqncase}
		}
	\qed
\end{proof}

\begin{lemma}[Stage Substitution]
	If $\G \V \mathcal{J}$, then $\G[\beta\mapsto B] \V \mathcal{J}[\beta\mapsto B]$.  Similarly, if $\V \G$, then $\V \G[\beta\mapsto B]$.
\end{lemma}

\begin{proof}
	Proved simultaneously by induction on derivations with case analysis on the last rule used.
	We show only representative cases.
		{
			\newcommand{\SB}{[\beta \mapsto B]}
			\newcommand{\GG}{\G\SB}
			\newcommand{\GGV}{\G\SB \V}

			\begin{rneqncase}{\TGen{}}{
					\mathcal{J} = \Lambda\alpha.M:\forall\alpha.\tau@A\\
					\G\V M:\tau@A & \alpha\notin\rm{FTV}(\G)\cup\rm{FTV}(A).
				}
				We can assume $\alpha \notin B$ because $\alpha$ is a bound variable.
				By the induction hypothesis, we have \(\G\SB\V M\SB:\tau\SB@A\).
				We can prove easily $\alpha \notin \FTV(\GG) \cup \FTV(A)$.
				Then, \(\GGV (\Lambda\alpha.M)\SB:(\forall\alpha.\tau)\SB@A\SB\) by \TGen.
			\end{rneqncase}

			\begin{rneqncase}{\KTW{}}{
					\mathcal{J} = \G\V \TW_\alpha \tau :: * @ A\\
					\G\V \tau :: * @ A\alpha.
				}
				\begin{itemize}
					\item If $\alpha \neq \beta$,\\
					      \( \GGV \tau\SB :: *\SB @ A\alpha\SB \) from the induction hypothesis.
					      \( \GGV \TW_\alpha \tau\SB :: *\SB @ A\SB \) from \KTW.

					\item If $\alpha = \beta$, \\
					      \( \GGV \tau\SB :: *\SB @ A\alpha\SB \) from the induction hypothesis and
					      it is identical to \( \GGV \tau\SB :: * @ AB \).
					      We can get \( \GGV \TW_B \tau\SB :: * @ A \) after repeatedly using \KTW{}.

				\end{itemize}
			\end{rneqncase}

			\begin{rneqncase}{\QGen{}}{
					\mathcal{J} = \Lambda\alpha.M \E \Lambda\alpha.N : \forall\alpha.\tau@A\\
					\G\V M\E N : \tau@A &  \alpha \notin \FTV(\G)\cup\FTV(A).
				}
				From the induction hypothesis, \( \GGV M\SB \E N\SB : \tau\SB@A\SB \).
				We can assume \(\alpha \notin B \) because \(\alpha\) is a bound variable.
				Thus, \( \alpha \notin \FTV(\G\SB)\cup\FTV(A\SB) \).
				\( \GGV \Lambda\alpha.M\SB \E \Lambda\alpha.N\SB : \tau\SB@A\SB \) from \QGen{} and
				it is identical to \( \GGV (\Lambda\alpha.M)\SB \E (\Lambda\alpha.N)\SB : \tau\SB@A\SB \).
			\end{rneqncase}
		}
	\qed
\end{proof}

\begin{lemma}[Agreement]
	\begin{enumerate}
		\item If \(\G\V \tau::K@A\), then \(\G\V K\iskind@A \).
		\item If \(\G\V M:\tau@A\), then \(\G\V \tau::*@A\).
		\item If \(\G\V K\E J@A\), then \(\G\V K\iskind@A\) and \(\G\V J\iskind@A\).
		\item If \(\G\V \tau\E \sigma :: K@A\), then \(\G\V \tau::K@A\) and \(\G\V \sigma::K@A\).
		\item If \(\G\V M\E N : \tau@A\), then \(\G\V M:\tau@A\) and \(\G\V N:\tau@A\).
	\end{enumerate}
\end{lemma}
\begin{proof}

	We can prove using induction on the derivation tree.
	We show some cases as examples.

	\begin{rneqncase}{\KCsp{}}{
			\G \V K :: * @A\alpha \text{ is derived from } \G \V K :: * @A.
		}
		By \WStar, \(\G\V *\iskind @A\alpha\).
	\end{rneqncase}

	\begin{rneqncase}{\TCsp{}}{
			\G\V \%_\alpha M :: \tau @A\alpha \text{ is derived from } \G\V M :: \tau @A
		}
		By the induction hypothesis, \( \G \V \tau :: * @ A \).
		By \KCsp, \( \G \V \tau :: * @ A\alpha \).
	\end{rneqncase}
	\begin{rneqncase}{\QBeta{}}{
			\G\V(\lambda x:\sigma.M)\ N\E M[x\mapsto N] : \tau[x \mapsto N]@A \text{ is derived from }\\
			\G,x:\sigma@A\V M:\tau@A \text{ and } \G\V N:\sigma@A.
		}
		From \TAbs{} and \TApp, \( \G\V(\lambda x:\sigma.M)\ N : \tau[x \mapsto N]@A \).
		From Term Substitution, \( \G\V M[x\mapsto N] : \tau[x \mapsto N]@A \).
	\end{rneqncase}
	\qed
\end{proof}

As we said in Section \ref{sec:properties}, we generalize Inversion Lemma to use induction.
\begin{lemma}[Inversion Lemma for $\Pi$ Type]
	\begin{enumerate}
		\item If $\G \V (\lambda x:\sigma.M) : \rho$, then there are $\sigma'$ and $\tau'$ such that
		      $\rho = \Pi x:\sigma'.\tau'$, $\G \V \sigma \E \sigma'@A$ and $\G ,x:\sigma'@A\V M:\tau'@A$.
		\item If $\G \V \rho \E (\Pi x:\sigma.\tau) : K @A$, then there are $\sigma', \tau', K$, and $J$ such that
		      $\rho = \Pi x:\sigma'.\tau'$, $\G \V \sigma \E \sigma' : K @A$, and $\G, x:\sigma@A\V \tau \E \tau' : J @A$.
		\item If $\G \V (\Pi x:\sigma.\tau) \E \rho : K @A$, then there are $\sigma', \tau', K$, and $J$ such that
		      $\rho = \Pi x:\sigma'.\tau'$, $\G \V \sigma \E \sigma' : K @A$, and $\G, x:\sigma@A\V \tau \E \tau' : J @A$.
	\end{enumerate}
\end{lemma}

\begin{proof}
	We can prove using induction on the derivation tree.
	We show a few main cases.

	\begin{rneqncase}{\TAbs{}}{
			\G\V \sigma::*@A \text{ and } \G,x:\sigma@A\V M:\tau@A.
		}
		We can take $\sigma$ and $\tau$ as $\sigma'$ and $\tau'$ in the statement.
	\end{rneqncase}

	\begin{rneqncase}{\TConv{}}{
			\G \V (\lambda x:\sigma.M) : \rho@A \text{ and } \G \V \rho \E (\Pi x:\sigma'.\tau)@A.
		}
		There are $\sigma'$ and $\tau'$ such that
		$\rho = \Pi x:\sigma'.\tau'$, $\G \V \sigma \E \sigma'@A$ and $\G ,x:\sigma'@A\V M:\tau'@A$
		by the induction hypothesis.
	\end{rneqncase}

	\begin{rneqncase}{\QTRefl{}}{}
		There two cases for the conclusion.
		\begin{itemize}
			\item If $\G \V \rho \E (\Pi x:\sigma.\tau) : K @A$,
			      we can use the third statement of the lemma as the induction hypothesis.
			\item If $\G \V (\Pi x:\sigma.\tau) \E \rho : K @A$,
			      we can use the second statement of the lemma as the induction hypothesis.
		\end{itemize}
	\end{rneqncase}
	\qed
\end{proof}

\begin{lemma}[Inversion Lemma for $\TW$ type]
	\begin{item}
	      \item If $\G \V \TB_\alpha M : \tau@A$, then
	      there is $\sigma$ such that $\tau = \TW_\alpha \sigma$ and $\G \V M : \sigma@A$.
	      \item If $\G \V \rho \E \TW_\alpha \tau : K @A$, then there are $\tau', K$, and $J$ such that
	      $\rho = \TW_\alpha \tau'$ and $\G \V \tau \E \tau' : K @A$.
	      \item If $\G \V \TW_\alpha \tau \E \rho : K @A$, then there are $\tau', K$, and $J$ such that
	      $\rho = \TW_\alpha \tau'$ and $\G \V \tau \E \tau' : K @A$.
	\end{item}
\end{lemma}

\begin{proof}
	We can prove using induction on the derivation tree.
	We show some cases as examples.

	\begin{rneqncase}{\TTB{}}{\G \V M : \sigma'@A\alpha.}
		We can take $\sigma'$ as $\sigma$.
	\end{rneqncase}
	\begin{rneqncase}{\TConv{}}{
			\G \V \TB_\alpha M : \tau'@A \text{ and } \G\V\tau' \E \tau :: K@A.
		}
		There is $\sigma$ such that $\tau' = \TW_\alpha \sigma$ and $\G \V M : \sigma@A$
		by using the induction hypothesis to \( \G \V \TB_\alpha M : \tau' @A\).
		There are $\sigma'$ and $K'$ such that $\tau = \TW_\alpha \sigma'$ and $\G \V \sigma \E \sigma' : K @ A$
		by using the induction hypothesis to \( \G \V \TW_\alpha \sigma : \tau :: K' @A\).
		Finally, $\G \V M : \sigma' @ A $ by \TConv.
	\end{rneqncase}
	\qed
\end{proof}

\begin{lemma}[Inversion for $\Lambda$ Type]
	\begin{item}
	      \item If $\G \V \Lambda\alpha.M : \tau$, then
	      there is $\sigma$ such that $\sigma = \forall\alpha.\sigma$ and $\G \V M : \sigma@A$.
	      \item If $\G \V \rho \E \forall\alpha.\tau : K @A$, then there are $\tau', K$ such that
	      $\rho = \forall\alpha.\tau'$ and $\G \V \tau \E \tau' : K @A$.
	      \item If $\G \V \forall\alpha.\tau \E \rho : K @A$, then there are $\tau', K$ such that
	      $\rho = \forall\alpha.\tau'$ and $\G \V \tau \E \tau' : K @A$.
	\end{item}
\end{lemma}

\begin{proof}
	We can prove using induction on the derivation tree.

	\qed
\end{proof}

\begin{lemma}[Inversion Lemma for Application]
	\begin{item}
	      \item If $\G \V (\lambda x:\sigma.M)\ N: \tau@A$ then there are $x$ and $\rho$ such that
	      $\G, x:\sigma \V M : \rho @A$ and $\G \V N:\sigma @ A$.
	\end{item}
\end{lemma}

\begin{proof}
	We can prove using induction on the derivation tree.
	\qed
\end{proof}

\begin{theorem}[Preservation for $\beta$-Reduction]
	If $\G\V M:\tau@A$ and $M \longrightarrow_{\beta} M'$, then $\G\V M':\tau@A$.
\end{theorem}

\begin{proof}
	{
		We can prove using induction on the type derivation tree.
		We show some cases as examples.
		\newcommand{\LB}{\longrightarrow_{\beta}}
		\begin{rneqncase}{\TApp{}}{}
			The shape of the reduction is one of following three.
			\begin{itemize}
				\item $(\lambda x:\sigma.N)\ L \LB N[x\mapsto L]$.\\
				      Because the last rule of the type derivation tree is \TApp,
				      we have $\G \V (\lambda x:\sigma.N) : (\Pi x:\sigma'.\tau')@A$ and
				      $\G \V L:\sigma' @A$.
				      By using Inversion Lemma for $\Pi$ type to $\G \V (\lambda x:\sigma.N) : (\Pi x:\sigma'.\tau')@A$,
				      we get $\G, x:\sigma \V N:\tau$ and $\G \V \sigma \E \sigma'$ and $\G ,x:\sigma \V \tau \E \tau'@A$.
				      By \TConv , $\G \V L:\sigma @A$.
				      By Term Substitution Lemma,
				      we get $\G \V N[x\mapsto L]:\tau[x\mapsto L]$.

				\item $M\ N \LB M'\ N$.\\
				      From the induction hypothesis and \TApp, the type is preserved for the reduction.
				\item $M\ N \LB M\ N'$.\\
				      From the induction hypothesis and \TApp, the type is preserved for the reduction.
			\end{itemize}
		\end{rneqncase}
	}
	\qed
\end{proof}

\begin{theorem}[Preservation for Term on $\blacklozenge$-Reduction]
	If $\G\V M:\tau@A$ and $M\longrightarrow_\blacklozenge N$, then $\G\V N:\tau@A$\\
\end{theorem}

\begin{proof}
	{
		We can prove using induction on the type derivation tree.
		We show the case of \TTBL{} as examples.
		Other cases are easy.
		\newcommand{\R}{\longrightarrow_{\blacklozenge}}
		\begin{rneqncase}{\TTBL{}}{}
			There are two cases for $\R$.
			\begin{itemize}
				\item $\TBL\TB M \R M$.\\
				      Because the last rule is \TTBL, we have $\G \V \TB M : \TW_\alpha \tau @A$.
				      By using Inversion Lemma for $\TW$ type to $\G \V \TB M : \TW_\alpha \tau @A$,
				      we get $\G \V M : \tau @A$
				\item $\TBL M \R \TBL M'$.\\
				      We can use the induction hypothesis and \TTBL.
			\end{itemize}
		\end{rneqncase}
		\qed
	}
\end{proof}

\begin{theorem}[Preservation for Term on $\Lambda$ Reduction]
	If $\G\V M:\tau@A$ and $M \longrightarrow_{\Lambda} N$, then $\G\V N:\tau@A$.
\end{theorem}

\begin{proof}
	{
		\newcommand{\R}{\longrightarrow_{\Lambda}}
		We can prove using induction on the type derivation tree.
		We show the case of \TIns{}{} as examples.
		Other cases are easy.
		\begin{rneqncase}{\TIns{}}{}
			There are two cases $\R$.
			\begin{itemize}
				\item $(\Lambda\alpha.M)\ B \R M[\alpha \mapsto B]$\\
				      Because the last rule is \TIns{}, we have $\G \V \Lambda\alpha.M\ : \forall\alpha.\tau @ A$.
				      By using Inversion Lemma for $\Lambda$ type to $\G \V \Lambda\alpha.M\ : \forall\alpha.\tau @ A$,
				      we obtain $\G \V M : \tau @ A$.
				      By using Stage Substitution Lemma to $\G \V M : \tau @ A$,
				      we obtain $\G[\alpha \mapsto B] \V M[\alpha \mapsto B] : \tau[\alpha \mapsto B] @ A[\alpha \mapsto B]$.
				      $\alpha$ is a bound variable therefore we can assume $\alpha \notin \FTV(\G) \cup \FTV(A)$.
				      So, we can rewrite $\G[\alpha \mapsto B] \V M[\alpha \mapsto B] : \tau[\alpha \mapsto B] @ A[\alpha \mapsto B]$ to
				      $\G \V M[\alpha \mapsto B] : \tau[\alpha \mapsto B] @ A$.
				\item $M\ B \R M'\ B$\\
				      By the induction hypothesis and \TIns.
			\end{itemize}
		\end{rneqncase}
		\qed
	}
\end{proof}

\begin{dfn}[$\natural$ translation]
	The $\natural$ translation is a translation from $\lambda^\text{MD}$ to $\lambda^\to$.
	\begin{itemize}
		\item Term
		      \begin{flalign*}
			      \natural(x) &= x & \\
			      \natural(\lambda x:\tau.M) &= \lambda x:\natural(\tau).\natural(M) & \\
			      \natural(M\ N) &= \natural(M)\ \natural(N)& \\
			      \natural(\TB_\alpha M) &= \natural(M) & \\
			      \natural(\TBL_\alpha M) &= \natural(M)& \\
			      \natural(\Lambda\alpha.M) &= \natural(M)& \\
			      \natural(M\ B) &= \natural(M) &
		      \end{flalign*}
		\item Type
		      \begin{flalign*}
			      \natural(X) &= X & \\
			      \natural(\Pi x:\tau.\sigma) &= \natural(\tau) \to \natural(\sigma) & \\
			      \natural(\tau\ x) &= \natural(\tau) & \\
			      \natural(\TW_\alpha \tau) &= \natural(\tau) & \\
			      \natural(\forall \alpha.\tau) &= \natural(\tau) &
		      \end{flalign*}
		\item Kind
		      \begin{flalign*}
			      \natural(K) &= * &
		      \end{flalign*}
		\item Context
		      \begin{flalign*}
			      \natural(\phi) &= \phi & \\
			      \natural(\G, x:T@A) &= \natural(\G), \natural(x):\natural(\tau) & \\
			      \natural(\G, X:K@A) &= \natural(\G) &
		      \end{flalign*}
	\end{itemize}
\end{dfn}

\begin{lemma}[Preservation of Equality in $\natural$]
	If $\G \V \tau \E \sigma @ A$ then $\natural(\tau) = \natural(\sigma)$.
\end{lemma}

\begin{proof}
	We can prove using induction on the type derivation tree.
	\qed\end{proof}

\begin{lemma}[Preservation of Typing in $\natural$]
	If $\G \V M:\tau@A$ in $\lambda^{\text{MD}}$ then $\natural(\G) \V \natural(M): \natural(\tau)$ in $\lambda^\to$.
\end{lemma}

\begin{proof}
	We can prove using induction on the type derivation tree.
	We show the case of \TApp{} and \TConv{} as examples.
	Other cases are easy.
	\begin{rneqncase}{\TApp{}}{
			\G \V M : (\Pi(x:\sigma).\tau) @ A \text{ and } \G \V N :\sigma @A
		}
		From the induction hypothesis, we have $\natural(\G) \V \natural(M) : \natural(\sigma) \to \natural(\tau)$ and $\natural(\G) \V \natural(N) : \natural(\sigma)$.
		Use the Application rule in $\lambda^\to$, we get $\natural(\G) \V \natural(M)\ \natural(N) : \natural(\tau)$.
		Because $\natural(M)\ \natural(N) = \natural(M\ N)$ from the definition of $\natural$, $\natural(\G) \V \natural(M\ N) : \natural(\tau)$ in $\lambda^\to$.
	\end{rneqncase}
	\begin{rneqncase}{\TConv{}}{
			\G\V M:\tau@A \text{ and } \G\V \tau\E\sigma@A
		}
		By using Preservation of typing in $\natural$ to $\G\V \tau\E\sigma@A$, we get $\natural(\tau) = \natural(\sigma)$.
		$\natural(\G) \V \natural(M):\natural(\tau)$ from the induction hypothesis.
		Then, $\natural(\G) \V \natural(M):\natural(\sigma)$.
	\end{rneqncase}
	\qed
\end{proof}

\begin{lemma}[Preservation of Substitution in $\natural$]
	If $\G, x:\sigma \V M:\tau@A$ and $\G \V N:\sigma@A$ in $\lambda^{\text{MD}}$
	then $\natural(M[x \mapsto N])$ = $\natural(M)[x\mapsto\natural(N)]$
\end{lemma}

\begin{proof}
	We prove by induction on the type derivation tree of $\G, x:\sigma \V M:\tau@A$.
	\qed\end{proof}

\begin{lemma}[Preservation of $\beta$-Reduction in $\natural$]
	If $\G \V M:\tau@A$ and $M \longrightarrow_\beta N$ in $\lambda^{\text{MD}}$
	then $\natural(M) \longrightarrow_\beta^+ \natural(N)$.
\end{lemma}

\begin{proof}
	{
		We prove by induction on the derivation of $\beta$ reduction of \LMD.
		We show only the main case.
		\newcommand{\R}{\longrightarrow_{\beta}}
		\begin{rneqncase}{$(\lambda x:\tau.M)\ N \R M[x \mapsto N]$}{}
			From the definition of $\natural$,
			$\natural((\lambda x:\tau.M)\ N)$ = $\lambda x:\natural(\tau).\natural(M)\ \natural(N)$.
			Because $\lambda x:\natural(\tau).\natural(M)\ \natural(N)$ is a typed term in $\lambda^\to$,
			we can do $\beta$ reduction from it.
			As a result of the reduction, we get $\natural(M)[x\mapsto\natural(N)]$.
			From Preservation of substitution in $\natural$, $\natural(M[x \mapsto N])$ = $\natural(M)[x\mapsto\natural(N)]$.
		\end{rneqncase}
		\qed
	}
\end{proof}

\begin{lemma}[Preservation of $\Lambda$-Reduction in $\natural$]
	If $\G \V M:\tau@A$ and $M \longrightarrow_\beta N$ in $\lambda^{\text{MD}}$
	then $\natural(M)$ =  $\natural(N)$.
\end{lemma}

\begin{proof}
	We prove by induction on the derivation of $\Lambda$-reduction of \LMD.
	We show only the main case.
	\begin{itemize}
		\item \textit{Case} \( (\Lambda\alpha.M)\ A \longrightarrow_\Lambda M[\alpha\mapsto A] \).\\
		      By the definition of $\natural$, \(\natural((\Lambda\alpha.M)\ A) = \natural(M)\).
		      Because \(\natural(M)\) does not contain \(\alpha\), \(\natural(M[\alpha\mapsto A]) = \natural(M)\).
	\end{itemize}
	\qed
\end{proof}

\begin{lemma}[Preservation of $\blacklozenge$-Reduction in $\natural$]
	If $\G \V M:\tau@A$ and $M \longrightarrow_\blacklozenge N$ in $\lambda^{\text{MD}}$
	then $\natural(M)$ =  $\natural(N)$.
\end{lemma}

\begin{proof}
	We prove by induction on the derivation of $\blacklozenge$-reduction of \LMD.
	We show only the main case.
	\begin{itemize}
		\item \textit{Case} \( \TBL_\alpha \TB_\alpha M \longrightarrow_\blacklozenge M \).\\
		      By the definition of $\natural$, \(\natural(\TBL_\alpha \TB_\alpha M) = \natural(M)\).
	\end{itemize}
	\qed
\end{proof}

\begin{theorem}[Strong Normalization]
	If $\G\V^A t:T$ then there is no infinite sequence of terms $(t_i)_{i\ge1}$ and $t_i \longrightarrow_{\beta, \TBL \TB,\Lambda} t_{i+1}$ for $i\ge 1$
\end{theorem}

\begin{proof}
	If there is an infinite reduction sequence in $\lambda^{\text{MD}}$
	then there are infinite beta reductions in the sequence.
	This is because reductions other than $\beta$-reduction reduce the size of a term.

	Then, we show Strong Normalization by proof by contradiction.
	We assume that there is an infinite reduction in a typed \LMD from typed term $M$.
	We can construct a typed term of simply typed lambda calculus $\natural(M)$ from Preservation of Typing in $\natural$ and
	$\natural(M)$ has infinite reductions from Preservation of Reduction in $\natural$.
	However, $\lambda^\to$ has a property of Strong Normalization, so there is no infinite reductions.
	\qed
\end{proof}

\begin{theorem}[Confluence(Church-Rosser Property)]
	Define $M \longrightarrow N$ as $M \longrightarrow_{\beta} N$ or $M\longrightarrow_\blacklozenge N$ or  $M \longrightarrow_{\Lambda} N$.
	For any term $M$, if $M \longrightarrow^* N$ and $M \longrightarrow^* L$,
	there exists $O$ that satisfies $N \longrightarrow^* O$ and $L \longrightarrow^* O$.
\end{theorem}

\begin{proof}
	Because we show the Strong Normalization of $\lambda^{\text{MD}}$, we can use Newman's lemma to prove Church-Rosser Property of $\lambda^{\text{MD}}$.
	Then, what we must show is Weak Church-Rosser Property.

	When we consider two different redexes in a $\lambda^{\text{MD}}$ term, they can only be disjoint, or one is a part of the other.
	In short, they never overlap each other.
	So, we can reduce one of them after we reduce another.
	\qed
\end{proof}

\begin{lemma}[Unique Decomposition]
	If $\G$ does not have any variable declared at stage $\varepsilon$
	and $\G \V M : \tau @ A$, then either
	\begin{enumerate}
		\item $ M \in V^A$, or
		\item $M$ can be uniquely decomposed into an evaluation context and a redex, that is, there uniquely exist $B, E^A_B$, and $R^B$ such that $M = E^A_B[R^B]$.
	\end{enumerate}
\end{lemma}

\begin{proof}
	{
	We prove by induction on the type derivation tree of $\G \V M:\tau@A$.

	\begin{rneqncase}{\TVar{}}{
			\G\V x:\tau@A \text{ is the root of the type derivation tree. }
		}
		We can assume $A \neq \varepsilon$ because $x:\tau@\varepsilon \notin \G$.
		Then, $x \in V^A$.
	\end{rneqncase}

	\begin{rneqncase}{\TTBL}{
			\G \V \TBL_\alpha M :\tau @ A\alpha \text{ is derived from } \G \V M : \TW_\alpha \tau @ A.
		}
		From the induction hypothesis, one of the following holds.
		\begin{enumerate}
			\item $ M \in V^A$.
			\item There is an unique triple of $(B, E^A_B, R^B)$ such that ($B = \varepsilon$ or $B = \beta$) and $M = E^A_B[R^B]$.
		\end{enumerate}
		\begin{itemize}
			\item If $ M \in V^A$ is true
			      \begin{itemize}
				      \item and $A=\varepsilon$, then\\
				            $ M = \TB_\alpha v^\alpha $ from Inversion Lemma and
				            $\TBL_\alpha \TB_\alpha v^\alpha = E^\alpha_\alpha [R^\alpha]$.
				      \item and $A\neq\varepsilon$, then\\
				            $\TBL_\alpha v^A \in V^{A\alpha}$.
			      \end{itemize}
			\item If there is an unique triple of $(B, E^A_B, R^B)$ such that ($B = \varepsilon$ or $B = \beta$) and $M = E^A_B[R^B]$
			      \begin{itemize}
				      \item and $ M = \TB_\alpha E^\alpha_B[R^B] $, then\\
				            $ \TBL_\alpha \TB_\alpha E^\alpha_B[R^B] \longrightarrow_s E^\alpha_B[R^B]$ doesn't hold because $ E^\alpha_B[R^B] \notin v^\alpha$.
				            So, given $B, E^A_B, R^B$ are the unique tuples satisfies the condition.
				      \item Otherwise,\\
				            It is obvious from the induction hypothesis and the definition of $E^A_B$.
			      \end{itemize}
		\end{itemize}

	\end{rneqncase}

	\begin{rneqncase}{\TIns}{
			\G \V M\ C :\tau[\alpha \mapsto C] @ A\\
			\text{ is derived from }  \G \V M : \forall\alpha.\tau @ A.
		}

		\begin{itemize}
			\item If $ A = \varepsilon$,\\
			      By the induction hypothesis, either
			      \begin{enumerate}
				      \item $ M \in V^\varepsilon$ or
				      \item there is a unique triple of $(B, E^\varepsilon_B, R^B)$ such that ($B = \varepsilon$ or $B = \beta$) and $M = E^\varepsilon_B[R^B]$.
			      \end{enumerate}

			      \begin{itemize}
				      \item If $ M \in V^\varepsilon$,\\
				            $ M = \Lambda\alpha.v^\varepsilon$ from Inversion Lemma.
				            Thus, $ \Lambda\alpha.v^\varepsilon\ C = E^\varepsilon_\varepsilon [R^\varepsilon]$
				      \item If there is an unique triple of $(B, E^\varepsilon_B, R^B)$ such that ($B = \varepsilon$ or $B = \beta$) and $M = E^\varepsilon_B[R^B]$,
				            we can decompose $E^\varepsilon_B[R^B]\ B$ uniquely
				            because $ E^\varepsilon_B[R^B] \neq \Lambda\alpha.v^\varepsilon$,
			      \end{itemize}

			\item If $ A \neq \varepsilon $,\\
			      By the induction hypothesis, either
			      \begin{enumerate}
				      \item $ M \in V^A$ or
				      \item there is a unique triple of $(B, E^A_B, R^B)$ such that ($B = \varepsilon$ or $B = \beta$) and $M = E^A_B[R^B]$.
			      \end{enumerate}
			      \begin{itemize}
				      \item If $ M \in V^A$,
				            it is clear that $v^A\ C \in V^A$.
				      \item If there are an unique triple of $(B, E^A_B, R^B)$ such that ($B = \varepsilon$ or $B = \beta$) and $M = E^A_B[R^B]$,
				            we can decompose $E^A_B[R^B]\ C$ uniquely because we cannot $\Lambda$ reduction at stage $A$.
			      \end{itemize}
		\end{itemize}
	\end{rneqncase}
	}
	\qed
\end{proof}

\begin{corollary}[Progress]
	If $x:\tau@\varepsilon \notin \G$ and $\G \V M : \tau @ A$ then $ M \in V^A $ or there is $M'$ such that $M \longrightarrow M'$.
\end{corollary}

\begin{proof}
	Immediate from Unique Decomposition.
	\qed\end{proof}

\fi
\end{document}